%% file: main.tex
\documentclass[letterpaper,hidelinks,10pt]{article}
\newif\ifdraft
\drafttrue
\usepackage[margin=1in]{geometry}
\usepackage{amsmath}
\usepackage{bbm}
\usepackage{amssymb,xspace}
\usepackage{amsthm}
\usepackage{enumerate}
\usepackage{xifthen}
\usepackage{mathtools}
\usepackage{algorithm}
\usepackage{algorithmicx}
\usepackage[noend]{algpseudocode}
\usepackage{hyperref}
\usepackage[nameinlink,capitalize]{cleveref}
\usepackage{thmtools,thm-restate}

\newtheorem{theorem}{Theorem}[section]

\newtheorem{lemma}[theorem]{Lemma}
\newtheorem{corollary}[theorem]{Corollary}

\theoremstyle{definition}
\newtheorem{definition}{Definition}[section]
\theoremstyle{definition}
\newtheorem*{remark}{Remark}
\newtheorem{observation}[theorem]{Observation}

\newcommand{\initC}{K}
\newcommand{\set}[1]{\left\{#1\right\}}
\newcommand{\ev}[2][]{\ifthenelse{\isempty{#1}}{\ensuremath{E\left[#2\right]}}{\ensuremath{E\left[#2 \mid #1\right]}}}
\newcommand{\E}[2][]{\ifthenelse{\isempty{#1}}{\ensuremath{E\left[#2\right]}}{\ensuremath{E\left[#2 \mid #1\right]}}}

\newcommand{\nh}{\ensuremath{\Gamma}}
\newcommand{\bigO}[1]{\ensuremath{O\left(#1\right)}}
\newcommand{\ind}{\ensuremath{\mathbbm{1}}}
\newcommand{\eps}{\ensuremath{\epsilon}}
\newcommand{\logstar}{\ensuremath{\log^*}}
\newcommand{\CONGEST}{\ensuremath{\mathsf{CONGEST}}\xspace}
\newcommand{\LOCAL}{\ensuremath{\mathsf{LOCAL}}\xspace}
\newcommand{\SETLOCAL}{\ensuremath{\mathsf{SET}\text{-}\mathsf{LOCAL}}\xspace}
\DeclareMathOperator{\polylog}{polylog}
\DeclareMathOperator{\poly}{poly}
\newcommand{\calM}{\mathcal{M}}
\newcommand{\calX}{\mathcal{X}}

\algnewcommand{\NoAlignComment}[1]{\hspace{1cm}\textit{// #1}}
\algnewcommand{\FullLineComment}[1]{\Statex\textit{// #1}}

\usepackage[textsize=tiny]{todonotes}

\newif\ifhideproofs
\hideproofsfalse 

\ifhideproofs
  \usepackage{environ}
  \NewEnviron{hide}{}

\fi

\newcommand{\myemail}[1]{\,$\cdot$\, {\small #1}}
\newcommand{\myaff}[1]{\,$\cdot$\, {\small #1}\par\smallskip}
\newenvironment{mycover}
{\list{}{\listparindent 0pt
        \itemindent    \listparindent
        \leftmargin    1cm
        \rightmargin   1cm
        \parsep        0pt}%
    \raggedright
    \item\relax}
{\endlist}
\newenvironment{myabstract}
{\list{}{\listparindent 1.5em%
        \itemindent    \listparindent
        \leftmargin    1cm
        \rightmargin   1cm
        \parsep        0pt}%
    \item\relax}
{\endlist}
\title{}
\begin{document}

\begin{mycover}
{\LARGE\bfseries\boldmath Efficient Deterministic Distributed Coloring with Small Bandwidth \par}
\bigskip
\bigskip
\bigskip

\textbf{Philipp Bamberger}
\myemail{philipp.bamberger@cs.uni-freiburg.de}
  \myaff{ University of Freiburg}

 \textbf{Fabian Kuhn}
  \myemail{ kuhn@cs.uni-freiburg.de}
	\myaff{University of Freiburg}

\textbf{Yannic Maus\footnote{Supported by the European Union's Horizon 2020 Research And  Innovation Programme under grant agreement no. 755839.}}
 \myemail{yannic.maus@campus.technion.ac.il}
\myaff{Technion}
\end{mycover}
\bigskip

\thispagestyle{empty}
\begin{myabstract}
\noindent\textbf{Abstract.}
  We show that the $(degree+1)$-list coloring problem can be solved
  deterministically in $O(D \cdot \log n \cdot\log^2\Delta)$ rounds in the
  \CONGEST model, where $D$ is the diameter of the graph, $n$ the
  number of nodes, and $\Delta$ the maximum degree.  Using the
  recent polylogarithmic-time deterministic network decomposition
  algorithm by Rozho\v{n} and Ghaffari [STOC 2020], this implies the
first efficient (i.e., $\poly\log n$-time) deterministic \CONGEST
algorithm for the $(\Delta+1)$-coloring and the $(\mathit{degree}+1)$-list
coloring problem. Previously the best known algorithm required
$2^{O(\sqrt{\log n})}$ rounds and was not based on network
decompositions.

Our techniques also lead to deterministic $(\mathit{degree}+1)$-list coloring algorithms for the
congested clique and the massively parallel computation (MPC) model. For the congested clique, we
obtain an algorithm with time complexity
$O(\log\Delta\cdot\log\log\Delta)$, for the MPC model, we obtain
algorithms with round complexity $O(\log^2\Delta)$ for the
linear-memory regime and $O(\log^2\Delta + \log n)$ for the sublinear
memory regime.
\end{myabstract}

\newpage
\setcounter{page}{1}
\input{sec1-intro}
\input{sec2-derand}
\input{sec3-application}
\input{sec4-CCMPC}
\input{sec5-MPCbasics}

	\bibliographystyle{alpha}
\bibliography{references}

\end{document}

%% file: sec1-intro.tex
\section{Introduction}

In the distributed message passing model, a communication network is abstracted as an $n$-node graph $G=(V,E)$. The nodes of $G$ host processors that communicate with each other through the edges of the graph. In the context of distributed graph algorithms, the objective is to solve some graph problem on $G$ by a distributed message passing algorithm. One of the most important and most extensively studied problems in the area is the  \emph{distributed graph coloring problem}, where we need to compute a proper (vertex) coloring of the communication graph $G$. Typically, at the beginning of an algorithm, the nodes of $G$ do not know anything about $G$, except maybe the names of their immediate neighbors and at the end of an algorithm, each node of $G$ needs to know its local part of the solution of the given graph problem (e.g., for distributed coloring, at the end, each node must know its own color). 
The two classic models in which distributed graph algorithms have been studied are the \LOCAL and the \CONGEST model~\cite{linial92,peleg2000distributed}. In both models, time is divided into synchronous rounds and in each round, each node can perform arbitrary internal computations and send a message to each of its neighbors. The time complexity of an algorithm is the number of rounds required for the algorithm to terminate. In the \LOCAL model, messages can be of arbitrary size, whereas in the \CONGEST model, all messages have to be of size $O(\log n)$ bits.

\paragraph{Distributed Graph Coloring.}
Computing a coloring with the optimal number of colors $\chi(G)$ was one of the first problems known to be NP-complete \cite{Karp1972}. In the distributed setting, one therefore aims for a more relaxed goal and the usual objective is to color a given graph $G$ with $\Delta+1$ colors, where $\Delta$ is the maximum degree of $G$ \cite{barenboimelkin_book}. Note that any graph admits such a coloring and it can be computed by a simple sequential greedy algorithm. Despite the simplicity of the sequential algorithm, the question of determining the distributed complexity of computing a $(\Delta+1)$-coloring has been an extremely challenging question. In particular, while $O(\log n)$-time randomized distributed $(\Delta+1)$-coloring algorithms have been known for more than 30 years, the question whether a similarly efficient (i.e., polylogarithmic time) \emph{deterministic} distributed coloring algorithm exists remained one of the most important open problems in the area \cite{barenboimelkin_book,stoc17_complexity} until it was resolved very recently by Rozho\v{n} and Ghaffari~\cite{RG19}. In \cite{RG19} the question was answered in the affirmative for the \LOCAL model by designing an efficient deterministic distributed method to decompose the communication graph into a logarithmic number of subgraphs consisting of connected components (clusters) of polylogarithmic (weak) diameter, a structure known as a \emph{network decomposition}~\cite{awerbuch89}. It was known before that an efficient deterministic algorithm for network decomposition would essentially imply efficient determinstic \LOCAL algorithms for all problems for which efficient randomized algorithms exist~\cite{linial92,awerbuch89,barenboimelkin_book,stoc17_complexity,FOCSDerand}.

Note that due to the unbounded message size any (solvable) problem can be solved in \emph{diameter} time in the \LOCAL model by simply collecting the whole graph topology at one node, solving the problem locally (potentially using unbounded computational power), and redistributing the solution to the nodes. Thus the small diameter components of a network decomposition almost immediately give rise to an efficient $(\Delta+1)$-coloring algorithm in the \LOCAL model. In light of the breakthrough by Rozho\v{n} and Ghaffari~\cite{RG19}, which implies that a polylogarithmic-time deterministic $(\Delta+1)$-coloring algorithm exists in the \LOCAL model, it is natural to ask whether a polylogarithmic-time deterministic algorithm also exists in the more restricted, but seemingly also more realistic \CONGEST model. As the main result of this paper, we answer this question in the affirmative.

\subsection{Our Contributions in \CONGEST}
The main technical contribution of this work is to provide an efficient deterministic \CONGEST model algorithm for the $(\Delta+1)$-coloring problem, and more generally the $(\mathit{degree}+1)$-list coloring problem in time proportional to the diameter of the graph. By using the network decomposition algorithm by Rozho\v{n} and Ghaffari~\cite{RG19}, the result also implies efficient \CONGEST algorithms for general graphs even if the diameter is large.

Given a graph $G=(V,E)$, a \emph{color space} $[C]$ and color lists $L(v)\subseteq [C]$ for each $v\in V$, a \emph{list-coloring} of $G$ is a proper $C$-coloring such that each node is colored with a color from its list. In this paper we consider the $(\mathit{degree}+1)$-list coloring problem, where the list sizes are $|L(v)|=\deg(v)+1$. Throughout the paper we assume that each color from each node's list fits in $O(1)$ messages in the \CONGEST model, i.e., we only consider list-coloring instances with $C=\poly n$. Whenever we do not explicitly mention the color space we assume $C=\poly\Delta$. Note that just as the $(\Delta+1)$-coloring problem the $(\mathit{degree}+1)$-list coloring problem admits a simple sequential greedy algorithm.

While in the \LOCAL model, the diameter of the graph is a trivial upper bound on the time needed to solve any graph problem, it is not clear per se whether a small diameter also helps for solving problems in the \CONGEST model. So far, there are a few examples where a small diameter helps in the \CONGEST model. In particular, the problems of computing a maximal independent set, a sparse spanner, and an $(1+\eps)\log \Delta$-approximation of a minimum dominating set can all be solved deterministically in time $D\cdot \polylog n$~\cite{CPS17,DISC18_DomSet,DKM19}. Other problems cannot profit from small diameter, e.g., verifying or computing a minimum spanning tree requires $\tilde{\Omega}(D+\sqrt{n})$ \CONGEST rounds~\cite{dassarma11} and solving many optimization problems exactly (minimum dominating set, vertex cover, chromatic number) or almost exactly (maximum independent set) requires $\tilde{\Omega}(n^2)$ \CONGEST rounds even if the diameter of the graph is constant \cite{CKP17,BCDELP19}.

\medskip

We prove the following theorem (when assuming $C\leq \poly\Delta$).

\begin{theorem}[simplified]\label{thm:diameterListsimple}
There is a deterministic \CONGEST algorithm that solves the list-coloring problem for instances $G=(V,E)$ with $|L(v)|\geq\deg(v)+1$ in time $\bigO{D\cdot \log n\cdot \log^2\Delta}$.
\end{theorem}
The algorithm to obtain \Cref{thm:diameterListsimple} is based on a similar basic strategy as the algorithms in \cite{CPS17,DISC18_DomSet,DKM19}, which achieve similar time complexities for other graph problems. 
More specifically, the proof of \Cref{thm:diameterListsimple} relies on designing a $0$-round randomized process that in expectation colors a constant fraction of the vertices. Then, we derandomize this process and iterate $O(\log n)$ times until all nodes are colored. 
In \Cref{ssec:nutshell} we elaborate on the challenges that occur with this approach and how we approach them. 

\medskip 
The most important implication of the result in \Cref{thm:diameterListsimple} is that it can be lifted to an efficient, i.e., $\polylog n$ time algorithm for general graphs, even for graphs with large diameter.  
\medskip

\begin{corollary}
\label{cor:effGraph}
There is a deterministic $\bigO{\log^8n}$-round \CONGEST algorithm for the $(\mathit{degree}+1)$-list-coloring problem.
\end{corollary}
The proof of \Cref{cor:effGraph} is based on iterating through the color classes of a \emph{suitable network decomposition} and applying \Cref{thm:diameterListsimple} on the clusters of the same color \cite{awerbuch89,RG19}. Many network decomposition algorithms, in particular the one in \cite{RG19}, only compute so-called \emph{weak-diameter network decompositions} in which the diameter of components is only small if edges and vertices outside the component can be used for communication. One needs additional care to use these decompositions in the \CONGEST model. For more details on the definition of a suitable network decomposition we refer to \Cref{sec:effCongest} and \cite{RG19}. 
We want to point out that improvements in computing such network decompositions immediately carry over to \Cref{cor:effGraph}.
\medskip

\Cref{cor:effGraph} is a drastic improvement over the state of the art even for the standard $(\Delta+1)$-coloring problem: Surprisingly until the beginning of 2019 the best deterministic \CONGEST algorithm for the $(\Delta+1)$-coloring problem was the $O(\Delta^{3/4}+\logstar n)$ algorithm by Barenboim \cite{barenboim15,barenboim18}. Even though the objective of \cite{barenboim15} was to optimize the runtime mainly as a function of the maximum degree $\Delta$, the paper also provided the fastest known algorithm if the runtime is solely expressed as a function of the number of nodes, i.e., it provided an $O(n^{3/4})$ round \CONGEST algorithm.  Very recently, the runtime for $(\mathit{degree}+1)$-list coloring was improved to $2^{O(\sqrt{\log \Delta})}\cdot\log n=2^{O(\sqrt{\log n})}$ rounds \cite{KuhnSoda20}. The algorithm in \cite{KuhnSoda20} does work in the \CONGEST model but it does not rely on network decompositions. Hence its runtime does not improve with the recent breakthrough result by Rozho\v{n} and Ghaffari \cite{RG19}.

\medskip
\subsection{Our Contributions in the CONGESTED CLIQUE and MPC}
The techniques of \Cref{thm:diameterListsimple} can also be adapted in an almost direct way to yield simple deterministic algorithms in the CONGESTED CLIQUE and the MPC model. By using the additional power in those models, we can speed up some parts of our derandomization. Because the main focus of our paper is on the \CONGEST model algorithm, the limited space in the main body of the paper is used for our \CONGEST results and all the proofs for the CONGESTED CLIQUE and the MPC model are moved to \Cref{sec:CC_MPC}.

\paragraph{The CONGESTED CLIQUE model \cite{Peleg03}:}
The CONGESTED CLIQUE model differs from the \CONGEST model in the way that the input graph might be different from the communication network. Given an input graph $G=(V,E)$ consisting of $n$ nodes, in each round, each node $u\in V$ can send a message of size $\bigO{\log n}$ to each other node $v\in V$ in the graph (i.e., although $G$ might be an arbitrary graph, the communication graph is a complete graph on the nodes $V$). Initially, each node only knows its neighbors in $G$. More specifically, we consider the UNICAST CONGESTED CLIQUE model, in which nodes are allowed to a send a different message to each other node in each round.

\smallskip

We prove the following result.

\begin{restatable}[CONGESTED CLIQUE]{theorem}{CongestedClique}\label{thm:CC}
There is a deterministic \text{CONGESTED CLIQUE} algorithm that solves the $(\mathit{degree}+1)$-list-coloring problem in time $O(\log\Delta\log\log\Delta)$.
\end{restatable}
The runtime of \Cref{thm:CC} is faster than the runtime of \Cref{thm:diameterListsimple} due to the following reasons.
First, due to all-to-all communication, we can avoid the diameter term in the runtime. Further, the $O(\log n)$ factor can be turned into a $O(\log\Delta)$ factor as we can send a subgraph to a single vertex of the clique and solve the problem locally as soon as the subgraph consists of at most $n/\Delta$ vertices which is achieved after $O(\log\Delta)$ rounds. In fact, this can even be achieved within $O(\log\log\Delta)$ rounds because with each phase in which a constant fraction of the nodes is colored, the routing capabilities compared to the number of remaining nodes is increased, accelerating the process of choosing a color for each node. Finally, $\Omega(\log \Delta)$ rounds of our derandomization procedure can be compressed into $O(1)$ rounds in the CONGESTED CLIQUE using the routing capabilities of the model.

We emphasize that while \Cref{thm:CC} provides the best known deterministic CONGESTED CLIQUE algorithm for $(\mathit{degree}+1)$-list-coloring, it does not improve the state of the art for the slightly weaker $(\Delta+1)$-list coloring problem, In \cite{parterCC18,parter19_arxivFix}, Parter provides a $O(\log\Delta)$-round deterministic algorithm CONGESTED CLIQUE algorithm for this problem.\footnote{A bug in the deterministic CONGESTED CLIQUE algorithm in the conference version \cite{parterCC18} has been fixed in \cite{parter19_arxivFix}.} Although \Cref{thm:CC} only improves the state of the art for a very special case, we still like to mention it, as the result almost immediately follows from our \CONGEST algorithm.

\paragraph{The MPC model \cite{KSV10,ANOY14}:} 
In the \emph{MPC model}, an input graph $G=(V,E)$ is distributed in a worst case manner among $M=\tilde{O}\left(\frac{|V|+|E|}{S}\right)$ machines, each having a memory of $S$ words. $S$ is a parameter of the model, a word consists of $O(\log |V|)$ bits and $\tilde{O}$ hides $\poly\log |V|$ factors that can be chosen arbitrarily by the designer of an algorithm. Time progresses in synchronous rounds in which machines exchange messages. In one round, every machine is allowed to send a (different) message to each other machine such that the size of all messages sent and received by a machine does not exceed its local memory. Additionally, each machine can perform an arbitrary local computation on its stored data. At the end of computation, each machine outputs a part of the solution which may not exceed its local memory. So for the coloring problem each machine is responsible for the output of some of the nodes and has to output their colors. The complexity of a deterministic MPC algorithm is the number of synchronous rounds until the problem is solved. 
We say that the MPC model works with \emph{linear memory} if $S=\tilde{O}(|V|)$ and with sublinear memory if there exists some constant $0<\alpha<1$ and $S=\tilde{O}(n^{\alpha})$.

\smallskip

The algorithm of \Cref{thm:CC} can also be implemented in the deterministic MPC model with linear memory, but without the additional improvement of turning the $\log\Delta$ into a $\log\log\Delta$ factor. However, additional care is needed to reason that none of the steps of the algorithm ever exceeds the memory of a machine, i.e., a machine never sends, receives nor stores more than $S=\tilde{O}(n)$ words. 
\begin{restatable}[MPC, linear memory]{theorem}{mpclin}\label{cor:MPClin}
There is a deterministic MPC algorithm that solves the $(\mathit{degree}+1)$-list-coloring problem in $O(\log^2 \Delta)$ rounds with linear memory.
\end{restatable}

For the sublinear memory regime we obtain the following result. The basic algorithm is the same as the one for the linear memory regime. Since in the sublinear memory regime, the neighborhood of a single node might not fit on a single machine, we have to use and adapt some standard MPC for handling operations like aggregation of a function over a single neighborhood or updating the color lists of all nodes after a partial coloring step (e.g., the last problem corresponds to designing an MPC algorithm to compute the intersection between two sets of size $N$, if machines with a local memory of $N^{\eps}$ are available). These basic MPC computations appear in \Cref{sec:MPCbasics}.

\begin{restatable}[MPC, sublinear memory]{theorem}{mpcsub}\label{cor:MPCsub}
There is a deterministic MPC algorithm that solves the $(\mathit{degree}+1)$-list-coloring problem in  $O(\log^2 \Delta+\log n)$ rounds with sublinear memory.
\end{restatable}

We remark that in independent work, Czumaj, Davies, and Parter~\cite{CDP19_MPCcoloring} have developed another deterministic $(\mathit{degree}+1)$-list coloring algorithm for the sublinear memory MPC model. The algorithm of \cite{CDP19_MPCcoloring} achieves a time complexity of  $O(\log\Delta+\log\log n)$ and it thus is significantly faster than our algorithm. In order to achieve this faster running time, the authors of \cite{CDP19_MPCcoloring} have to make the assumption that even if the graph is sparse and consists of only $n^{1+o(1)}$ or even $\tilde{O}(n)$ edges, the MPC algorithm needs at total memory of $\Omega(n^{1+\alpha})$ for some constant $\alpha>0$. In our results, the total memory is of order $\tilde{O}(m+n)$, if $m$ denotes the number of edges of $G$. We also believe that our algorithm is of interest because in our opinion, both our algorithm and the analysis are substantially simpler than the algorithm and analysis of \cite{CDP19_MPCcoloring}.

\subsection{Related Work}
We continue with a short overview on the most important distributed graph coloring algorithms in the \LOCAL and \CONGEST model. For a more detailed discussion of previous work on distributed graph coloring we refer to \cite{barenboimelkin_book,KuhnSoda20} for deterministic algorithms and \cite{chang18_coloring} for randomized algorithms. 

\paragraph{Deterministic Distributed Coloring:} 
The work on distributed coloring started more than 30 years ago with work on several symmetry breaking problems in the parallel setting \cite{alon86,cole86,goldberg88,luby86} and a seminal paper by Linial \cite{linial92} that introduced the \LOCAL model for solving graph problems in a distributed setting. Linial showed that a graph can be colored (deterministically) in $O(\log^* N)$ rounds with $O(\Delta^2)$ colors and that $\Omega(\log^* N)$ rounds are necessary even for graphs of degree $2$ where $N$ is the size of the identifier space.\footnote{The function $\logstar x$ denotes the number of iterated logarithms needed to obtain a value at most $1$, that is, $\forall x\leq 1: \logstar x=0,\ \forall x>1: \logstar x = 1 + \logstar\log x$. As most results in the area assume that the identifier space $N$ is of size $\poly(n)$ Linial's lower bound and upper bounds are usually used as $\Omega(\logstar n)$ and $O(\logstar n)$; we do the same.}
 Even though these algorithms and lower bounds were devised for the \LOCAL model, all of them also directly apply to the \CONGEST model.  If randomization is allowed, already the early works imply that a $(\Delta+1)$-coloring can always be computed in $O(\log n)$ rounds, even in the \CONGEST model~\cite{alon86,luby86,linial92,johansson99}.

The focus of our current work is to express the complexity of distributed graph coloring in the \CONGEST model as a function of $n$, i.e, the number of nodes of the network.   Despite ample attention to the problem in the \LOCAL model, the $2^{O(\sqrt{\log n})}$-time solution of \cite{panconesi95} has been the state of the art until the recent breakthrough by Rozho\v{n} and Ghaffari~\cite{RG19}, which provided the first $\polylog n$-time algorithm.   However, to the best of our knowledge and somewhat surprisingly there are no published results on the $(\Delta+1)$-coloring problem in the \CONGEST model, when runtime is expressed solely as a function of $n$. But there has been extensive
work on determining the problem's round complexity in terms of the maximum degree $\Delta$ and most of these algorithms also work in \CONGEST: A  simple single-round color elimination scheme combined with Linial's $O(\Delta^2)$-coloring algorithm provided an $O(\Delta\log \Delta+\logstar n)$ algorithm \cite{szegedy93,Kuhn2006On}. This was improved to $O(\Delta+\logstar n)$ rounds by using a divide-and-conquer approach based on defective colorings \cite{barenboim09,spaa09}. 
The first and only sublinear in $\Delta$ algorithm in the for $(\Delta+1)$-coloring was obtained by Barenboim and used $O(\Delta^{3/4}+\logstar n)$ rounds~\cite{barenboim15,barenboim18}. 
The current state of the art in the \LOCAL model uses $\tilde{O}(\sqrt{\Delta}+\logstar n)$ rounds, but it does not extend to \CONGEST \cite{fraigniaud16,barenboim18}. 
 There are faster algorithms if the final coloring can use more than $\Delta+1$ colors: \cite{barenboim10} shows that coloring with $O(\Delta^{1+\eps})\gg\Delta+1$ colors for some constant $\eps>0$ can be done in $O(\log \Delta\log n)$ rounds. If one desires to reduce the allowed communication from $O(\log n)$ bits per round to a single bit per round, Barenboim, Elkin and Goldenberg provided algorithm  that uses $O(\Delta+\log n+ \logstar n)$ rounds and colors with $O(\Delta)$ colors \cite{barenboim18}. In its original version the  algorithm computes an $O(\Delta)$-coloring in $O(\Delta+\logstar n)$ rounds in the standard \CONGEST model.

While there has been extensive progress on upper bounds, the original $\Omega(\logstar n)$ lower bound by Linial is still the best known lower bound on the $(\Delta+1)$-coloring problem. It was recently shown that in a weak
version of the \LOCAL model (called the \SETLOCAL model), the
$(\Delta+1)$-coloring problem has a lower bound of
$\Omega(\Delta^{1/3})$ \cite{disc16_coloring}.

\paragraph{Randomized Distributed Coloring:} There has also been
a lot of work on understanding the randomized complexity of the
distributed coloring problem
\cite{kothapalli06,SchneiderW10,barenboim12,PettieS13,HarrisSS16,podc18_Deltacoloring,chang18_coloring}. Most of the work focuses on the \LOCAL model and a 
particularly important contribution was provided by Barenboim, Elkin,
Pettie, and Schneider \cite{barenboim12}, who introduced the so-called
graph shattering technique to the theory of distributed graph
algorithms. The paper shows that a $(\Delta+1)$-coloring can be
computed in time $O(\log\Delta) + 2^{O(\sqrt{\log\log n})}$ in the \LOCAL model. 
\cite{GhaffariSODA19} showed that the same can be done in the \CONGEST model and in both algorithms the $2^{O(\sqrt{\log\log n})}$ term can most likely be reduced with the result from \cite{RG19}.
The state of the art for randomized $(\Delta+1)$-coloring in the \LOCAL model is the graph shattering based $\poly\log\log n$ algorithm by Chang et al.~\cite{chang18_coloring,RG19} that does not translate to the \CONGEST model. 

\paragraph{$(\Delta+1)$-Coloring in the CONGESTED CLIQUE and MPC:}
In the special case of $\Delta =O(n^{1/3})$ the MIS algorithm in \cite{CPS17} combined with a well-known reduction from the $(\Delta+1)$-coloring problem to the MIS problem \cite{luby86,linial92} can be used to deterministically compute a $(\Delta + 1)$-coloring in $O(\log \Delta)$  rounds. 
Building up on the result from \cite{CPS17} Parter provided an $O(\log \Delta)$ deterministic algorithm \cite{parterCC18}.
All other results in the CONGESTED CLIQUE and the MPC model are randomized: \cite{CFGUZ19} provides a $O(1)$-round algorithm in the CONGESTED CLIQUE and a $O(\sqrt{\log\log n}$) round randomized algorithm for the MPC model with sublinear memory. 
Previously Assadi et al.~\cite{assadi19} provided an $O(1)$ randomized algorithm in the MPC model in the linear memory regime, i.e., with memory $\tilde{O}(n)$. 

\paragraph{Derandmization in the \CONGEST Model:}
As pointed out earlier \cite{CPS17,DISC18_DomSet,DKM19} use a similar derandomization strategy for other graph problems. \cite{CPS17} computes an MIS in $\tilde{O}(D)$ rounds and uses the same strategy to obtain deterministic algorithms for spanners. \cite{DISC18_DomSet} computes a $O(\log^2 n)$ approximation for the minimum dominating set problems and \cite{DKM19} computes a $(1+\eps)\log \Delta$-approximation for the minimum dominating set problem for any constant $\eps>0$.  We cite \cite{CPS17}:
\textit{
\begin{center}
This work  opens a window to many additional intriguing questions. First, we would like to see many more local problems being derandomized despite congestion restrictions.
\end{center}
}
Despite the fact that \cite{CPS17,DISC18_DomSet,DKM19} and the current paper use the same derandomization strategy, applying it to further graph problems is unfortunately non-black-box.
\smallskip

Kawarabayashi and Schwartzman \cite{schwartzmanAdapting18} also use the concept of derandomization in the \CONGEST model. Their approach relies on iterating through the color classes of a given coloring; by only considering a suitably chosen subgraphs and defective colorings  they obtain several algorithms for various cut problems including a deterministic $O(\eps^{-2}\log \Delta+\logstar n)$ round algorithm for an $(1/2-\eps)$-approximation of the \emph{max cut problem}.

\subsection{Our Derandomization in a Nutshell} 
\label{ssec:nutshell}

\paragraph{General derandomization strategy.} We use the general derandomization strategy as recently introduced by Censor-Hillel, Parter, and Schwartzman in \cite{CPS17} and afterwards also used in \cite{DISC18_DomSet,DKM19}. We start with a simple and efficient randomized algorithm that is guaranteed to make good progress in expectation. By using the method of conditional expectations (see e.g., \cite[Chapter 6.3]{Mitzenmacher05}), we then turn this randomized algorithm into a deterministic \CONGEST algorithm with a time complexity of $D\cdot\polylog n$. Finally, in combination with the new polylogarithmic-time network decomposition algorithm of \cite{RG19}, the running time of the algorithm can be reduced from $D\cdot\polylog n$ to only $\polylog n$.

A bit more specifically, assume that we have a random variable $X_v$ that measures the progress of each node $v\in V$ in our given randomized algorithm such that $X:=\sum_{v\in V} X_v$ measures the global progress of the algorithm. Assume for example that a small value of $X$ implies fast progress, whereas a large $X$ implies slow progress (in the context of coloring $X$ could for example measure the number of conflicts in a single random coloring step). Assume further that we can reasonably upper bound the expected value $\E{X}$ of $X$ even if the random choices of the nodes in the randomized algorithm are only $k$-wise independent for some $k=\polylog n$ (i.e., for any subset of $k$ nodes, the choices are independent). It is well-known that a set of $n$ $\polylog n$-wise independent random bits can be generated from only $\polylog n$ independent random bits~\cite{Vadhan12}. The algorithm can therefore be implemented with a shared random seed of $\polylog n$ bits and we can use the method of conditional expectations to \emph{deterministically} find an assignment to these $\polylog n$ shared bits for which the value of $X$ is not larger than its expectation. When implementing the algorithm in the \CONGEST model, we deterministically fix the $\polylog n$ bits of the random seed one-by-one, where fixing a single bit involves a global aggregation for computing the conditional expectations of $\E{X}$ based on setting the bit to $0$ or $1$. One bit of the random seed can therefore be fixed in $O(D)$ time in the \CONGEST model.

\paragraph{Computing a partial coloring.} 
The arguably most natural distributed random coloring algorithm is to let each node choose a color from its list uniformly at random \cite{johansson99}. For the expected number of conflicts $X_v$ of a node $v$ we have
\begin{equation}\label{eq:potentialtry1}
\E{X_v} = \sum_{u\in \nh{v}} \frac{|L(u)\cap L(v)|}{|L(u)|\cdot|L(v)|} \leq \sum_{u\in\nh{v}}\frac{|L(u)|}{|L(u)|\cdot|L(v)|}=\deg(v)\cdot \frac{1}{|L(v)|} < 1~.
\end{equation}
Note that for \eqref{eq:potentialtry1} to hold, we need only pairwise independent choices of the colors. It follows that the expected number of conflicts is $\sum_{v\in V}\E{X_v}<n$. The derandomization by the method of conditional expectations would then choose a color for each node such that $\sum_{v\in V}X_v\leq\sum_{v\in V}\E{X_v}<n$, i.e., the number of conflicts is upper bounded by its expectation which is at most $n$. It follows that at least half of the nodes have $X_v\leq1$, i.e., at most one neighbor that has chosen the same color and thus at least half of these nodes can keep their color. The whole process can then be repeated for $O(\log n)$ iterations to color all vertices. However, in order to compute the conditional expectation  of $X_v$, $v$ needs to know the lists of its neighbors which is too expensive to acquire in the \CONGEST model. 
\paragraph{Computing Conditional Expectations.} Therefore we choose another approach which is inspired by the algorithm of \cite{KuhnSoda20}: Let the global color space be $[C]$ for some integer $C$ (i.e., for each $v\in V$, we have $L(v)\subseteq [C]$). Each color thus has a binary representation of $\log C$ bits. We define a process where we fix each node's color bit-by-bit, i.e., we run $\log C$ phases where in each phase, each node fixes the next bit of its color. Each time a node $v$ has fixed a bit, its list $L(v)$ is reduced to the subset of colors which have the sequence of bits that we have fixed so far as prefix and as soon as all bits have been fixed, each node has picked a candidate color. We show that if the color prefixes are extended with the correct probabilities, this process colors a constant fraction of the vertices in expectation. One can also view it as a slowed down version of directly taking a color from the list uniformly at random. Now we apply the aforementioned derandomization technique only to the zero-round algorithm of choosing a single bit of the color prefix. The benefit of this approach is that it admits a pessimistic estimator of \Cref{eq:potentialtry1} that allows the efficient computation of conditional expectations.

\paragraph{Shorter Random Seeds.} In \cite{CPS17,DISC18_DomSet,DKM19} the length of the shared random seed is $\polylog n$. As the seed length appears as a factor in the runtime of the resulting deterministic algorithms, one wishes to keep it as short as possible. In our algorithms we manage to reduce the seed length to $O(\log \Delta+\log\log C)$ bits; in particular the seed length is independent of $n$. The main ingredient for a shorter seed length is the observation that pairwise independent random coins for adjacent nodes are sufficient in our randomized algorithms and these coins can be produced from a random seed whose length does not depend on $n$. This observation might be helpful for derandomizing other algorithms.


%% file: sec2-derand.tex
\section{Degree+1 List Coloring in Diameter Time}
\label{sec:diameterList}
Throughout, let $D$ denote the diameter of a graph. Often we run algorithms on subgraphs of a graph; however, all our algorithms can be implemented such that $D$ always refers to the diameter of the original communication graph. For $C\in\mathbb{N}$ we introduce the notation $[C]:=\{0,\dots,C-1\}$. 
Given a graph $G=(V,E)$, a \emph{color space} $[C]$ and \emph{color lists} $L(v)\subseteq [C]$ for each $v\in V$, a \emph{list-coloring} $\phi:V\to [C]$ of $G$ is a proper $C$-coloring such that each node is colored with a color from its list, i.e, for any edge $\{u,v\}\in E$ we have $\phi(u)\neq\phi(v)$ and for all $v\in V$ we have $\phi(v)\in L(v)$. In this paper we consider the $(\mathit{degree}+1)$-list-coloring problem, where $|L(v)|=\deg(v)+1$ for each $v\in V$.

\smallskip

As our main technical contribution, we show that we can list-color a constant fraction of the nodes of a graph.
\begin{lemma}\label{thm:colorhalf}
There is a deterministic \CONGEST algorithm that given a list-coloring instance $G=(V,E)$ with color space $[C]$, lists $L(v)\subseteq[C]$ for which $|L(v)|\geq\deg(v)+1$ holds for all $v\in V$ and an initial $\initC$-coloring of $G$, list-colors a $1/8$ fraction of the nodes in $\bigO{D\cdot \log C\cdot (\log\Delta+\log \initC+\log\log C)}$ rounds.

When the result is applied to a subgraph of a communication graph $G$, $\deg(v)$ refers to the degree of $v$ in the subgraph and $\Delta$ to the maximum degree of the subgraph, but the diameter $D$ refers to the diameter of $G$.
\end{lemma}

Before we explain the algorithm of \Cref{thm:colorhalf} we show that it directly implies the following theorem.

\medskip

\noindent\textbf{\Cref{thm:diameterListsimple}}
\textit{There is a deterministic \CONGEST algorithm that solves the list-coloring problem for instances $G=(V,E)$ with $L(v)\subseteq[C]$ and $|L(v)|\geq\deg(v)+1$ in time $\bigO{D\cdot\log n\cdot\log C\cdot(\log\Delta+\log\log C)}$.}

\begin{proof}
First we use Linial's coloring algorithm (\cite{linial92}) to compute a $\initC=O(\Delta^2)$ coloring of the graph in $O(\logstar n)$ rounds. This coloring is used as the input coloring to $O(\log n)$ iterations of \Cref{thm:colorhalf}. In each iteration $i=0,\dots,O(\log n)$ we color a constant fraction of the still uncolored vertices and after the iteration each still uncolored vertex removes the colors of newly colored neighbors from its list. As a node only removes a color from its list if also its `uncolored degree' drops, the residual graph $G_i$ that is induced by the still uncolored nodes after iteration $i$ forms a feasible instance for \Cref{thm:colorhalf}, i.e., it satisfies the condition $|L(v)|\geq\deg_{G_i}(v)+1$ for all $v\in V(G_i)$. Thus after $O(\log n)$ iterations each vertex is colored with a color from its initial list and the total runtime is 
\begin{align*}
\bigO{D\cdot\log n\cdot\log C\cdot(\log\Delta+\log \initC+ \log\log C)}=\bigO{D\cdot\log n\cdot\log C\cdot(\log\Delta+\log\log C)} & \qedhere
\end{align*}
\end{proof}
In the introduction we presented a simplified version of \Cref{thm:diameterListsimple} by setting $C=\poly\Delta$~.

\begin{remark}
When applying \Cref{thm:diameterListsimple} to disconnected subgraphs the diameter term in the runtime can be substituted by the maximum diameter of the connected components.
\end{remark}

We now continue by explaining the ideas of the algorithm of \Cref{thm:colorhalf}; the formal proof of \Cref{thm:colorhalf} appears at the very end of \Cref{sec:diameterList}. Note that the given initial $\initC$-coloring is only used for symmetry breaking purposes and does not relate to the lists of the nodes.

\paragraph{General Algorithmic Idea:} In the algorithm of \Cref{thm:colorhalf}, every node deterministically (and tentatively) selects a color from its list such that a constant fraction of nodes can permanently keep their selected color. We now describe this selection process in more detail: Given a list-coloring instance $G=(V,E)$ with color space $[C]$, each color is represented by a bitstring of length $\lceil\log C\rceil$ (e.g., color $2$ is represented as the bit string $0\dots010$ and not just as the bit string $10$). Our algorithm operates in $\lceil\log C\rceil$ phases and in each phase we determine one further bit of each node's color. That is, each $u\in V$ maintains a bitstring $s(u)$ with the property that $s(u)$ is the prefix of some color(s) in $L(u)$, starting with $s(u)=s_0(u)$ being the empty string and successively extending the prefix $s(u)$ by one bit per phase. We write $s_{\ell}(u)$ for the bitstring that node $u$ has chosen after phase $\ell$. The string $s_{\lceil\log C\rceil}(u)$ corresponds to the tentative color that $u$ selects.
For $\ell=1,\ldots,\lceil\log C\rceil$ we define
\[L_{\ell}(u):=\{x\in L(u) \mid \text{ prefix of $x$ is $s_{\ell}(u)$}\}\]
 as the set of $u$'s \emph{candidate colors} at the end of phase $\ell$,  i.e, the set of colors in $L(u)$ that start with $s_{\ell}(u)$. We note that our algorithm chooses all prefixes $s_{\ell}(u)$ in such a way that $L_{\ell}(u)$ is always non-empty for all $\ell$ and $u\in V$.
We further define the remaining \emph{conflict graph} after iteration $\ell$ as
\begin{align*}G_{\ell}:=(V,E_{\ell}), \text{ where } E_{\ell}:=\{\{u,v\}\in E\mid s_{\ell}(u)=s_{\ell}(v)\}~,\end{align*}
and the remaining \emph{conflict degree} $\deg_{\ell}(u):=\deg_{G_{\ell}}(u)$. That is, in graph $G_{\ell}$ any edge of $G$ is considered as deleted as soon as its nodes choose different prefixes.
To ensure that a constant fraction of the vertices can keep their selected color we need a suitable measure of progress for the extension of prefixes, that is, we need a potential function that captures the \emph{'usefulness'} of the prefixes for our purpose. Informally, the potential relates (but does not equal!) to the expected number of monochromatic edges if each node chooses a random candidate color from its list, i.e., a random color that is consistent with its current prefix. We define the following \emph{\textbf{potential function}} 

\[\Phi_{\ell}(u):=\frac{\deg_{\ell}(u)}{|L_{\ell}(u)|}~.\]
At the beginning, when all prefixes are empty, we have $\Phi_0(u)<1$ for all $u\in V$ and thus $\sum_{u\in V}\Phi_0(u)<n$. We will give an algorithm that extends all prefixes bitwise while keeping the overall increase of the potential small such that when all $\ell=\lceil\log C\rceil$ bits are fixed, we still have $\sum_{u\in V}\Phi_{\ell}(u)\leq 2n$. It follows that a constant fraction of the nodes $u$ have $\Phi_{\ell}(u)<4$.  At this stage all nodes have selected a single candidate color, i.e., $|L_{\ell}(u)|=1$ for all $u\in V$. Thus $\Phi_{\ell}(u)$ equals the number of neighbors $v$ of $u$ with the same candidate color. Hence a constant fraction of the vertices has at most $3$ neighbors that have the same selected color which is sufficient to permanently color a constant fraction of the vertices of $G$.

\medskip

In \Cref{sec:randomized} we show that using biased coin flips to determine the next bit of the prefix yields a $0$-round randomized algorithm that, \textit{in expectation}, has no increase in the potential even if the nodes only use pairwise independent coin flips (\Cref{lem:expectation}). In the same section we show that the expected increase of the potential is small if the probabilities of the coin flips are chosen slightly inaccurate (\Cref{lem:expectslack}). 

In \Cref{sec:derand} we show that such biased coin flips can be deterministically produced from a short random seed (\Cref{lem:coinProd}). Thus, over the randomness of the seed, the expected potential increase is small. In the following we use the method of conditional expectation and a BFS tree to derandomize this algorithm, i.e., we \emph{deterministically} pick a \emph{good} random seed that incurs only a small increase in the potential (\Cref{lem:derand}).

\medskip

\subsection{Extending Prefixes: Zero Round Randomized Algorithms}\label{sec:randomized}
To fix bit $\ell$ of all prefixes (that we describe by \Cref{alg:randbit}) each node flips a coin to determine its $\ell$-th bit. The bit and also the coin equals $1$ with probability $p_u:=\frac{k_1(u)}{|L_{\ell-1}(u)|}$ where $L_{\ell-1}(u)$ is the list of current candidate colors and  $k_1(u):=|\{x\in L_{\ell-1}(u) \mid x[\ell]=1\}|$ is the number of candidate colors whose $\ell$-th bit equals $1$. 
Thus $p_u$ is the fraction of candidate colors whose $\ell$-th bit equals $1$. This process can be seen as a slowed down version of selecting a color from the initial list uniformly at random as iterating this process for $\lceil\log C\rceil$ times yields the same probability for each color to be selected. However, we do not know how to immediately derandomize the non slowed down process.

\begin{algorithm}[H]
\caption{Randomized One Bit Prefix Extension}
\begin{algorithmic}[0]
\State Input: Bitstring $s_{\ell-1}(u)$ of length $\ell-1$ for all $u\in V$
\For{each node $u$ in parallel}
\State $p_u:=\frac{k_1(u)}{|L_{\ell-1}(u)|}$ where $k_1(u):=|\{x\in L_{\ell-1}(u)\mid x[\ell]=1\}|$
\State Coin Flip: Set $s_{\ell}(u)=s_{\ell-1}(u)\circ1$ with probability $p_u$ and $s_{\ell}(u)=s_{\ell-1}(u)\circ0$ otherwise

($\circ$ represents the concatenation of strings)
\EndFor
\end{algorithmic}
\label{alg:randbit}
\end{algorithm}

\begin{lemma}\label{lem:expectation}
Let $\ell\leq\lceil\log C\rceil$ and assume the prefixes $s_{\ell-1}(v)$ are fixed for all nodes $v$. Let all nodes choose the $\ell$-th bit according to \Cref{alg:randbit}. Then we obtain
\begin{align}
\ev{\sum_{v\in V}\Phi_{\ell}(v)}\leq\sum_{v\in V}\Phi_{\ell-1}(v)
\end{align} 
if the coin flips in \Cref{alg:randbit} of adjacent nodes are independent. Furthermore, the list of candidate colors of each node never becomes empty.
\end{lemma}

\begin{proof}
We can write the sum of the potentials of all nodes as \[\sum_{v\in V}\Phi_{\ell}(v)=\sum_{\{u,v\}\in E_{\ell}}\left(\frac{1}{|L_{\ell}(u)|}+\frac{1}{|L_{\ell}(v)|}\right)~.\]
For each edge $e=\{u,v\}\in E_{\ell-1}$ we introduce a random variable \[X_e=\ind_{e\in E_{\ell}}\left(\frac{1}{|L_{\ell}(u)|}+\frac{1}{|L_{\ell}(v)|}\right)\]
which can be seen as the contribution of $e$ to the potential. Edge $e$ survives the $\ell$-th phase if either both endpoints choose 1 as their $\ell$-th bit, which happens with probability $p_up_v$ due to the pairwise independence of the bit-choice, or if both choose 0 as their $\ell$-th bit, which happens with probability $(1-p_u)(1-p_v)$. In the first case we obtain $|L_{\ell}(u)|=p_u|L_{\ell-1}(u)|$ and $|L_{\ell}(v)|=p_u|L_{\ell-1}(v)|$ and in the second case $|L_{\ell}(u)|=(1-p_u)|L_{\ell-1}(u)|$ and $|L_{\ell}(v)|=(1-p_u)|L_{\ell-1}(v)|$. To compute $\ev{X_e}$ we hence define

\begin{align*}
&A:=\left\{
\begin{array}{ll}
p_up_v\left(\frac{1}{p_u|L_{\ell-1}(u)|}+\frac{1}{p_v|L_{\ell-1}(v)|}\right) & \text{if }p_u,p_v>0 \\
0 & \, \text{else} \\
\end{array}
\right.\\
&B:=\left\{
\begin{array}{ll}
(1-p_u)(1-p_v)\left(\frac{1}{(1-p_u)|L_{\ell-1}(u)|}+\frac{1}{(1-p_v)|L_{\ell-1}(v)|}\right) & \text{if }p_u,p_v<1 \\
0 & \, \text{else} \\
\end{array}
\right.
\end{align*}
and obtain
\begin{align*}
\ev{X_e}&=A+B\leq p_up_v\left(\frac{1}{p_u|L_{\ell-1}(u)|}+\frac{1}{p_v|L_{\ell-1}(v)|}\right)+(1-p_u)(1-p_v)\left(\frac{1}{(1-p_u)|L_{\ell-1}(u)|}+\frac{1}{(1-p_v)|L_{\ell-1}(v)|}\right)\\
&=\frac{p_v}{|L_{\ell-1}(u)|}+\frac{p_u}{|L_{\ell-1}(v)|}+\frac{(1-p_v)}{|L_{\ell-1}(u)|}+\frac{(1-p_u)}{|L_{\ell-1}(v)|}=\frac{1}{|L_{\ell-1}(u)|}+\frac{1}{|L_{\ell-1}(v)|}~.
\end{align*}
It follows
\[\ev{\sum_{v\in V}\Phi_{\ell}(v)}=\ev{\sum_{e\in E_{\ell-1}}X_e}\leq\sum_{\{u,v\}\in E_{\ell-1}}\left(\frac{1}{|L_{\ell-1}(u)|}+\frac{1}{|L_{\ell-1}(v)|}\right)=\sum_{v\in V}\Phi_{\ell-1}(v)~.\]
The lists of candidate colors never becomes empty as vertices only choose to extend their prefix by $0$ (or $1$ respectively) if they also have a candidate color with the respective extension; for this property to hold we merely need that the probabilities $p_v\in\{0,1\}$ are exactly represented. 
\end{proof}

We will not immediately derandomize the described procedure but a very similar one where nodes produce their biased random coins to represent the probabilities $p_v$ from a common short random seed. Then, in the derandomization process, we will find a \emph{good} seed from which nodes can determine the values of their coins. Producing biased coins from the common random seed implies that not all probabilities can be produced, in fact, instead of having a coin that equals 1 with an arbitrary probability $p_v$ we can only produce probabilities of the type $k/2^b$ for some large enough $b$ that we will chose later. That is, each $p_v$ can only be approximated up to some $\varepsilon=\Theta(2^{-b})$. The next lemma shows that the expected increase of the potential can be kept small with these inaccurate probabilities. 

\begin{lemma}\label{lem:expectslack}
Let $\ell\leq\lceil\log C\rceil$ and assume the prefixes $s_{\ell-1}(v)$ are fixed for all nodes $v$. Let all nodes choose the $\ell$-th bit according to \Cref{alg:randbit}, but with the following adjustment: For each node $v$, if $p_v=0$ or $p_v=1$, then $v$ chooses 1 as its $\ell$-th bit with probability $0$ or $1$ respectively. For all other values of $p_v$, $v$ chooses 1 as its $\ell$-th bit with some probability in the interval $[\max\{0,p_v-\varepsilon\},\min\{1,p_v+\varepsilon\}]$ for some $0\leq\varepsilon<1$.
Then we obtain
\begin{align}
\ev{\sum_{v\in V}\Phi_{\ell}(v)}\leq\sum_{v\in V}\Phi_{\ell-1}(v)+10\varepsilon\Delta n
\end{align} 
if the coin flips of adjacent nodes are independent. Furthermore, the list of candidate colors of each node never becomes empty.
\end{lemma}

\begin{proof}
For an edge $e=\{u,v\}$ we define $X_e$ as in \Cref{lem:expectation} and

\begin{align*}
&A:=\left\{
\begin{array}{ll}
(p_u+\varepsilon)(p_v+\varepsilon)\left(\frac{1}{p_u|L_{\ell-1}(u)|}+\frac{1}{p_v|L_{\ell-1}(v)|}\right) & \text{if }p_u,p_v>0 \\
0 & \, \text{else} \\
\end{array}
\right.\\
&B:=\left\{
\begin{array}{ll}
(1-p_u+\varepsilon)(1-p_v+\varepsilon)\left(\frac{1}{(1-p_u)|L_{\ell-1}(u)|}+\frac{1}{(1-p_v)|L_{\ell-1}(v)|}\right) & \text{if }p_u,p_v<1 \\
0 & \, \text{else} \\
\end{array}
\right.
\end{align*}
We obtain
\[A\leq\frac{p_v}{|L_{\ell-1}(u)|}+\frac{p_u}{|L_{\ell-1}(v)|}+\frac{\varepsilon}{|L_{\ell-1}(u)|}+\frac{\varepsilon}{|L_{\ell-1}(v)|}+\frac{\varepsilon p_v}{p_u|L_{\ell-1}(u)|}+\frac{\varepsilon p_u}{p_v|L_{\ell-1}(v)|}+\frac{\varepsilon^2}{p_u|L_{\ell-1}(u)|}+\frac{\varepsilon^2}{p_v|L_{\ell-1}(v)|}~.\]

\smallskip

As we have $p_u\geq1/|L_{\ell-1}(u)|$ and $p_v\geq1/|L_{\ell-1}(v)|$ if $p_u,p_v>0$ (cf. the definition of $p_u$ in \Cref{alg:randbit}), we obtain
\[A\leq\frac{p_v}{|L_{\ell-1}(u)|}+\frac{p_u}{|L_{\ell-1}(v)|}+\frac{\varepsilon}{|L_{\ell-1}(u)|}+\frac{\varepsilon}{|L_{\ell-1}(v)|}+\varepsilon p_v+\varepsilon p_u+2\varepsilon^2\]
and analogously
\[B\leq\frac{1-p_v}{|L_{\ell-1}(u)|}+\frac{1-p_u}{|L_{\ell-1}(v)|}+\frac{\varepsilon}{|L_{\ell-1}(u)|}+\frac{\varepsilon}{|L_{\ell-1}(v)|}+\varepsilon(1-p_v)+\varepsilon(1-p_u)+2\varepsilon^2~.\]
It follows
\[\ev{X_e}=A+B\leq\frac{1}{|L_{\ell-1}(u)|}+\frac{1}{|L_{\ell-1}(v)|}+\frac{2\varepsilon}{|L_{\ell-1}(u)|}+\frac{2\varepsilon}{|L_{\ell-1}(v)|}+2\varepsilon+4\varepsilon^2\leq\frac{1}{|L_{\ell-1}(u)|}+\frac{1}{|L_{\ell-1}(v)|}+10\varepsilon\]
and thus
\[\ev{\sum_{v\in V}\Phi_{\ell}(v)}=\ev{\sum_{e\in E_{\ell-1}}X_e}\leq\sum_{\{u,v\}\in E_{\ell-1}}\left(\frac{1}{|L_{\ell-1}(u)|}+\frac{1}{|L_{\ell-1}(v)|}+10\varepsilon\right)\leq\sum_{v\in V}\Phi_{\ell-1}(v)+10\varepsilon\Delta n~.\]
Because nodes can represent the probabilities 0 and 1 exactly, the proof that the list of candidate colors of each node never becomes empty goes along similar lines as in \Cref{lem:expectation}.
\end{proof}

\subsection{Extending Prefixes Deterministically through Derandomization}\label{sec:derand}

To fix the $\ell$-th bit of all prefixes deterministically  we produce the nodes' coins for \Cref{alg:randbit} from a short random seed, such that (1) the coins of each two adjacent nodes are independent, (2) the coins can represent the probabilities $p_v$ with a sufficient accuracy and (3) the common random seed is short enough to find a good seed deterministically and efficiently in the CONGEST model. To this end, we need the following result on how to compute biased coins from a random seed.

\begin{definition}[\cite{Vadhan12}]
For $N,M,k\in\mathbb{N}$ such that $k\leq N$, a family of functions $\mathcal{H}=\{h:[N]\to [M]\}$ is $k$-wise independent if for all distinct $x_1,\dots,x_k\in[N]$, the random variables $h(x_1),\dots,h(x_k)$ are independent and uniformly distributed in $[M]$ when $h$ is chosen uniformly at random from $\mathcal{H}$.
\end{definition}

\begin{theorem}[\cite{Vadhan12}]\label{thm:seed}
For every $a,b,k$, there is a family of $k$-wise independent hash functions $\mathcal{H}=\{h:\{0,1\}^a\to\{0,1\}^b\}$ such that choosing a random function
from $\mathcal{H}$ takes $k\cdot\max\{a,b\}$ random bits.
\end{theorem}
We can use \Cref{thm:seed} to produce biased random coins for the vertices of a graph such that adjacent vertices' coins are independent. 
\begin{lemma}\label{lem:coinProd}
Given a graph $G=(V,E)$, an integer $b>0$, probabilities $(p_v)_{v\in V}$ and a $\initC$-coloring $\psi:V\rightarrow [\initC]$ one can efficiently compute  random coins $(C_v,p_v)_{v\in V}$ from a seed of length $2\max\{\log K,b\}$ with the following properties: 
\begin{itemize}
\item $C_v$ equals $1$ with probability $p_v\pm 2^{-b}$ if $p_v\notin \{0,1\}$
\item $C_v$ equals $1$ with probability $p_v$ if $p_v\in \{0,1\}$
\item the coins of adjacent vertices are independent. 
\end{itemize}
\end{lemma}
\begin{proof}
Set $a=\log \initC$. By \Cref{thm:seed}, one can efficiently select a function $h_S:[\initC]\rightarrow [2^b]$ from a uniformly chosen random seed $S$ of length $2\cdot\max\{a,b\}=O(\max\{\log \initC, b\})$. Here, for all $i\in [\initC]$ the random variable $h_S(i)$ (over the randomness of the random seed) is uniformly distributed in $[2^b]$; further the random variables $h_S(0),\ldots,h_S(\initC-1)$ are pairwise independent. We obtain the desired biased coins, i.e., random variables over the randomness of the seed by defining
\begin{align}\label{eq:coins}
 C_v = \left\{\begin{array}{lr}
       1, & \text{ if } \frac{h_S(\psi(v))}{2^b}<p_v\\
        0, & \text{otherwise}
        \end{array}\right.
\end{align}
As $h_S(\psi(v))$ only assumes values in $[2^b]$, we always have $\frac{h_S(\psi(v))}{2^b}<1$ and never $\frac{h_S(\psi(v))}{2^b}<0$. Hence, if $p_v=0$ or $p_v=1$, then $C_v$ equals $1$ with probability $0$ or $1$, respectively. Generally, as $h_S(\psi(v))$ is uniformly distributed in $[2^b]$, we have $\Pr(C_v=1)=i/2^b$ with $i=|\{k\in [2^b]\mid k/2^b<p_v\}|$. That is, $\Pr(C_v=1)$ equals $p_v$ rounded up to the next multiple of $1/2^b$, i.e.,
\[p_v\leq\Pr(C_v=1)\leq p_v+2^{-b}~.\]

Note that although we use the same random variable $h_S(i)$ for all nodes that have input color $i$ in $\psi$, the probabilities of their coins are not equal. However, two adjacent vertices have distinct input colors $i\neq j$ for which we use the independent random variables $h_S(i)$ and $h_S(j)$ and hence their coins are independent.
\end{proof}

The next lemma shows how to use the method of conditional expectation to find a good random seed that only incurs a small increase of the potential.

\begin{lemma}\label{lem:derand}
There is a deterministic CONGEST algorithm that given a $\initC$-coloring of the graph, fixes the $\ell$-th bit of all prefixes in time $\bigO{D\cdot(\log\initC+\log\Delta+\log\log C)}$ such that 
\begin{align}\label{eq:lemmaClaim}\sum_{u\in V}\Phi_{\ell}(u)\leq\sum_{u\in V}\Phi_{\ell-1}(u)+\frac{n}{\lceil\log C\rceil}
\end{align}
and the list of candidate colors of each node does not become empty for $\ell\leq\lceil\log C\rceil$.
\end{lemma}

\begin{proof}
Assume that prefixes of length $\ell-1$ are already fixed. 
For each $v\in V$ we define $p_v:=k_1(v)/|L_{\ell-1}(v)|$ and apply \Cref{lem:coinProd} with $b=\lceil\log 10\Delta\lceil\log C\rceil\rceil$ to obtain biased coins $\{(C_v,p_v)\mid v\in V\}$ from a seed of length $d=O(\log\initC+\log\Delta+\log\log C)$ that are independent for adjacent nodes of $G$ and where $C_v$ equals $1$ with probability $p_v\pm 2^{-b}$ if $p_v\notin\{0,1\}$ and where $C_v$ equals $1$ with probability $0$ or $1$ if $p_v$ equals $0$ or $1$, respectively. 
Define a variant of \Cref{alg:randbit} executed with the coins $\{(C_v,p_v)\mid v\in V\}$, that is, node $v$ fixes the $\ell$-th bit of its prefix to the value of $C_v$. Note that this algorithm per se is not a distributed algorithm but uses shared randomness in the form of the shared random seed. 
The described random process satisfies all properties needed to apply \Cref{lem:expectslack}, with $\varepsilon=\frac{1}{2^b}=\frac{1}{10\Delta\lceil\log C\rceil}$ and we can bound the expected increase of the sum of all potentials by
\begin{align}
\ev{\sum_{v\in V}\Phi_{\ell}(v)}\leq\sum_{v\in V}\Phi_{\ell-1}(v)+\frac{n}{\lceil\log C\rceil}~. \label{eq:ev}
\end{align} 
Next, we derandomize this process with the method of conditional expectation to find a seed that only incurs a small increase of the potential.

\paragraph{Derandomization:} For $j\in\{1,\dots,d\}$, let $R_j$ be the random variable that describes the value of the $j$-th bit of the random seed. To derandomize the aforementioned algorithm we iterate through the bits of the random seed and deterministically fix a \emph{good bit} $r_j$ for each $R_j$ using the method of conditional expectation---we will later define the notion of a \emph{good bit}. The computed \emph{good seed} $s=r_1,\ldots,r_d$ will be such that the potential does not increase by much if coins are flipped and prefixes are extended according to the seed $s$. We use a BFS tree with a designated root as a leader that gathers all the necessary information to find and distribute good bits for the seed. We obtain a deterministic distributed algorithm to extend the prefixes by one bit without increasing the potential by much. The runtime to find a single good bit for some $R_j$ will be $O(D)$ due to communication over the BFS tree. The total runtime of extending prefixes by one bit is $O(D\cdot \text{seedlength})=O(D\cdot d)=O(D\cdot(\log\initC+\log\Delta+\log\log C))$.

\smallskip
	
\paragraph{Finding a good bit for $R_j$:} Assume we already chose good values $R_1=r_1,\dots,R_{j-1}=r_{j-1}$ for a $1\leq j\leq b$ and we want to find a good bit $r_j$ for $R_j$. By the law of total expectation there must be an $r_j\in\{0,1\}$ such that

\begin{align}
\ev[R_1=r_1,\dots,R_j=r_j]{\sum_{v\in V}\Phi_{\ell}(v)}\leq\ev[R_1=r_1,\dots,R_{j-1}=r_{j-1}]{\sum_{v\in V}\Phi_{\ell}(v)}~, \label{eq:derand}
\end{align}
where the randomness is over the non-determined random bits of $R_{j+1},\ldots, R_d$ and $R_{j},\ldots,R_d$ of the random seed, respectively. We call a value $r_j$ satisfying \Cref{eq:derand} a \emph{good bit}; note that the property of $r_j$ being good depends on the choice of $r_1,\ldots,r_{j-1}$.

To let the leader of the BFS tree find a good bit $r_j$, each node $v\in V$ computes
\begin{align*}
x_v^0&=\ev[R_1=r_1,\dots,R_{j-1}=r_{j-1}, R_j=0]{\Phi_{\ell}(v)}\text{\quad and\quad } x_v^1&=\ev[R_1=r_1,\dots,R_{j-1}=r_{j-1}, R_j=1]{\Phi_{\ell}(v)}~.
\end{align*}
In order to compute $x_u^0$ and $x_u^1$, $u$ needs to know $\Pr(\{u,v\}\in G_{\ell}\mid R_1=r_1,\dots,R_{j-1}=r_{j-1},R_j=i)$ for $i\in\{0,1\}$ for all its neighbors $v$ in $G_{\ell-1}$. The event $\{u,v\}\in G_{\ell}$ only depends on the values of $C_u$ and $C_v$. Therefore, $u$ learns $k_1(v)$ (the number of colors in $L_{\ell-1}(v)$ with 1 as the $\ell$-th bit) and $\psi(v)$ from all its neighbors $v$ in $G_{\ell-1}$ in one CONGEST round. This is sufficient to compute $C_v$ for any seed $S$ as the initial color $\psi(v)$ determines which random variable $v$ uses to produce $C_v$ and all nodes know how to generate the random variable $h_S(t)$ for any color $t\in[K]$ from $S$. Note that each node $u$ is aware of its neighbors in $G_{\ell-1}$ when nodes always exchange the latest chosen bit of their prefix each time a new prefix bit is fixed.

\smallskip

Next step, we aggregate $\sum_{v\in V}x_v^0$ and $\sum_{v\in V}x_v^1$ at the leader, choose

\begin{align*}
r_j:=\underset{i\in\{0,1\}}{\arg\min}\sum_{v\in V}x_v^i
\end{align*}
and broadcast $r_j$ to all nodes. $r_j$ is a good bit, as we know there is a bit which fulfills (\ref{eq:derand}) and $r_j$ is chosen as the bit which minimizes the left hand side in (\ref{eq:derand}).

\medskip

After $d$ iterations, we found a good seed $r_1\dots,r_d$. By iteratively applying (\ref{eq:derand}) we obtain

\begin{align}\label{eq:noRandomness}
\ev[R_1=r_1,\dots,R_d=r_d]{\sum_{v\in V}\Phi_{\ell}(v)}\leq\ev{\sum_{v\in V}\Phi_{\ell}(v)}\stackrel{(\ref{eq:ev})}{\leq}\sum_{v\in V}\Phi_{\ell-1}(v)+\frac{n}{\lceil\log C\rceil}~.
\end{align}
As we have fixed the complete random seed as $s=r_1\circ \ldots \circ r_d$, the left hand side of \Cref{eq:noRandomness} does not contain any randomness. 
The claim of the lemma (\Cref{eq:lemmaClaim}) follows if each node deterministically extends its prefix by one bit according to $s$.

In the same way as in \Cref{lem:expectslack} we can show that by applying our randomized algorithm, the list of candidate colors of each node does not become empty (this holds for \emph{any} possible outcome of the algorithm, i.e., for any choice of the random seed). Hence, this also holds for our deterministic algorithm as its output equals one possible outcome of the randomized one.
\end{proof}

\begin{proof}[Proof of \Cref{thm:colorhalf}]
Initially, for each node $v$ we have $|L(v)|\geq\deg(v)+1$ and thus $\sum_{v\in V}\Phi_0(v)=\sum_{v\in V}\frac{\deg(v)}{|L(v)|}\leq n$.
We apply the algorithm from \Cref{lem:derand} $\ell=\lceil\log C\rceil$ times and obtain

\[\sum_{v\in V}\Phi_{\ell}(v)\leq\sum_{v\in V}\Phi_0(v)+\frac{n}{\lceil\log C\rceil}\cdot\lceil\log C\rceil\leq2n~.\]
It follows that at least half of the nodes $u$ have $\Phi_{\ell}(u)<4$. Denote the set of these nodes by $V_{<4}$. As all nodes have chosen prefixes of length $\ell=\lceil\log C\rceil$ and all lists never become empty we obtain $|L(u)|=1$ for all $u\in V$. Hence for any $u\in V$ the value $\Phi_{\ell}(u)=\frac{\deg_{\ell}(u)}{|L(u)|}=\deg_{\ell}(u)$ equals the number of neighbors that have selected the same candidate color, i.e., the number of neighbors $v$ with $s_{\ell}(u)=s_{\ell}(v)$. So the graph $G_{\ell}[V_{<4}]$ has maximum degree at most $\Delta_{\ell}=3$. We compute an MIS on $G_{\ell}[V_{<4}]$ in $O(\logstar K)$ rounds: The given $K$-coloring of $G$ induces a $K$-coloring of $G_{\ell}[V_{<4}]$ which we can transform to an $O(\Delta_{\ell}^2)=O(1)$ coloring in $O(\logstar K)$ rounds with Linial's algorithm \cite{linial92}. Then we compute an MIS on $G_{\ell}[V_{<4}]$ by iterating through the color classes. Each node $u\in V_{<4}$ that is contained in the MIS colors itself permanently with $s_{\ell}(u)$; all nodes not in the MIS forget their candidate color and remain uncolored. Due to the maximum degree of $G_{\ell}[V_{<4}]$ the MIS has size at least $|V_{<4}|/4\geq n/8$, that is, at least a $1/8$ fraction of all nodes get colored.

The runtime equals the time needed for $\lceil\log C\rceil$ iterations of the algorithm from \Cref{lem:derand} (plus $O(\log^*\initC)$ rounds for computing the MIS), i.e., $\bigO{\log C\cdot D\cdot(\log\initC+\log\Delta+\log\log C)}$.
\end{proof}


%% file: sec3-application.tex
\section{Efficient $(\mathit{degree}+1)$-List Coloring in $\mathsf{CONGEST}$} \label{sec:effCongest}
For solving the $(\mathit{degree}+1)$-list-coloring problem deterministically in $\poly\log n$ rounds in \CONGEST, we first compute an $(\alpha,\beta)$-network decomposition with $\alpha,\beta=\poly\log n$ (see \Cref{def:nd}), that is, a decomposition of the network graph into $\poly\log n$ color classes such that connected components (\emph{clusters}) in each color class have \emph{small diameter}, i.e., diameter $\poly\log n$. Then we iterate through the $\poly\log n$ classes and apply the algorithm from \Cref{thm:diameterListsimple} on the small diameter clusters of one class in parallel. The concept of network decomposition was introduced in \cite{awerbuch89} and was later differentiated into weak and strong decompositions \cite{LS93}. However, both concepts are not suitable for our purpose. In a weak decomposition, when solving the list-coloring on a cluster, it might be necessary to communicate also via edges outside the cluster. Hence, it is not possible to run the algorithm from \Cref{thm:diameterListsimple} on $\omega(\log n)$ clusters in parallel as there might be an edge being involved in all these computations. In contrast, a strong decomposition would be sufficient, however, we do not know how to compute it in $\poly\log n$ rounds in \CONGEST.
 We therefore introduce a slightly more general definition of a network decomposition, which also includes a congestion parameter $\kappa$. A similar definition was for example used previously in \cite{GK19}.

\begin{definition}[Network decomposition with congestion]
\label{def:nd}
An $(\alpha,\beta)$-network decomposition with congestion $\kappa$ of a graph $G=(V,E)$ is a partition of $V$ into clusters $C_1,\dots,C_p$ together with associated subtrees $T_1,\ldots,T_p$ of $G$ and a color $\gamma_i\in\{1,\dots,\alpha\}$ for each cluster $C_i$ such that
\begin{enumerate}[(i)]
\item the tree $T_i$ of cluster $C_i$ contains all nodes of $C_i$ (but it might contain other nodes as well)
\item each tree $T_i$ has diameter at most $\beta$
\item clusters that are connected by an edge of $G$ are assigned different colors
\item each edge of $G$ is contained in at most $\kappa$ trees of the same color
\end{enumerate}
\end{definition}

When we assume to have a network decomposition on a graph, we require that each node knows the color of the cluster it belongs to and for each of its incident edges $e$ the set of associated trees $e$ is contained in. Note that a decomposition according to this definition has weak diameter $\beta$ and a strong network decomposition is a decomposition with congestion 1 where the tree $T_i$ of each cluster $C_i$ contains exactly the nodes in $C_i$.

\begin{theorem}[\cite{RG19}]
\label{thm:rozhon}
There is a deterministic algorithm that computes an $\left(\bigO{\log n},\bigO{\log^3n}\right)$-network decomposition with congestion $\bigO{\log n}$ in $\bigO{\log^8n}$ rounds in the \CONGEST model.\footnote{In on-going unpublished work \cite{RG19personal} it is shown that diameter and runtime can be improved. These improvements carry over to our results.}
\end{theorem}
We use \Cref{thm:rozhon} and \Cref{thm:diameterListsimple} to list-color graphs efficiently in the \CONGEST model.

\medskip

\noindent\textbf{\Cref{cor:effGraph}}
\textit{There is a deterministic \CONGEST algorithm that solves the list-coloring problem for instances $G=(V,E)$ with $L(v)\subseteq[C]$ and $|L(v)|\geq\deg(v)+1$ in $\bigO{\log^8n}$ rounds.}

\begin{proof}
Given a graph $G$, we compute an $\left(\bigO{\log n},\bigO{\log^3n}\right)$-network decomposition with congestion $\bigO{\log n}$ in $O(\log^8 n)$ rounds  using the algorithm from \Cref{thm:rozhon}. To list-color $G$ we iterate through the color classes of the network decomposition. When handling a single color class we need to solve a list-coloring problem on each cluster which is done by applying \Cref{thm:diameterListsimple} on all clusters of the color class in parallel. We continue with a detailed description and runtime analysis of the process: Let $L_G(v)$ be the initial list of $v\in V$ in $G$. Assume all clusters with colors $1,\dots,k-1$ are already list-colored. To list-color the clusters of color $k$, every node in such a cluster $\mathcal{C}$ updates its list, i.e., it deletes all colors from its list $L_G(v)$ taken by already colored neighbors in $G$ and obtains a new list $L_{\mathcal{C}}(v)$. Let $\deg_{\mathcal{C}}(v)$ denote the degree of $v$ in the graph induced by the nodes in cluster $\mathcal{C}$ and $\deg_G(v)$ the degree of $v$ in $G$. We obtain $|L_{\mathcal{C}}(v)|\geq\deg_{\mathcal{C}}(v)+1$ because initially we have $|L_G(v)|\geq\deg_G(v)+1=\deg_{\mathcal{C}}(v)+|\nh_G(v)\setminus \mathcal{C}|+1$ and we remove at most one color from $|L_G(v)|$ for each neighbor in $|\nh_G(v)\setminus \mathcal{C}|$. 

Thus we can apply the algorithm from \Cref{thm:diameterListsimple} for each cluster $\mathcal{C}$ of color $k$ in parallel where all aggregation and broadcast in the derandomization part is done over the associated tree of the cluster. This tree has diameter $O(\log^3n)$ (note that $n$ is the size of $G$ and not of the tree). Messages over edges that are contained in more than one tree are pipelined and as each edge of $G$ is contained in at most $\bigO{\log n}$ trees of the same color, the number of rounds to list-color all clusters of color $k$ is at most $\bigO{\log n}$ times the number of rounds required to list-color a single cluster (without pipelining) which is $\bigO{\log^4n\cdot\log C\cdot(\log\Delta+\log\log C)}$ according to \Cref{thm:diameterListsimple}. So in total $\bigO{\log^5n\cdot\log C\cdot(\log\Delta+\log\log C)}$ rounds are needed to list-color a single color class of the network decomposition. After iterating through all $\bigO{\log n}$ colors, each node chose a color from its list and has no conflict with a neighbor. Hence we computed a $(\mathit{degree}+1)$-list-coloring of $G$ in time $\bigO{\log^8n+\log^6n\cdot\log C\cdot(\log\Delta+\log\log C)}=\bigO{\log^8n}$.
\end{proof}


%% file: sec4-CCMPC.tex
\section{List Coloring in the CONGESTED CLIQUE and MPC}
\label{sec:CC_MPC}

In this section we show how to adapt the algorithm from \Cref{thm:diameterListsimple} to obtain faster algorithms in the CONGESTED CLIQUE and MPC model. As in these models we have direct communication between the nodes/machines, the derandomization from \Cref{thm:colorhalf} can be sped up in three ways: (1) Communication with a leader can be done directly, removing the dependence on the diameter in the runtime; (2) due to the large routing capabilities of the respective models, we can derandomize $\log n$ bits of the seed in constant time (as described in the following proofs), reducing the influence of the seed-length in the runtime to a constant factor. Sometimes we can even change the base of the bit representation of the colors to get further speedups; (3) once the graph is sparse enough we can send it to one machine (node) and solve it locally in that machine. Thus the $\log n$-factor that is due to coloring a constant fraction of the vertices in each iteration can be turned into a $\log \Delta$-factor.

In essence, using the algorithm from \Cref{sec:diameterList} but choosing several bits of the random seeds at once and gradually adapting the base of the bit representation of the colors to the current routing capabilities---the fewer uncolored vertices remain, the larger the routing capabilities become compared to the number of uncolored nodes---imply the following result. 

\medskip 

\noindent\textbf{\Cref{thm:CC}}
\textit{There is a deterministic \text{CONGESTED CLIQUE} algorithm that solves the $(\mathit{degree}+1)$-list-coloring problem for instances $G=(V,E)$ with $L(v)\subseteq[C]$ and $|L(v)|\geq\deg(v)+1$ in time $O(\log\log\Delta\log C)$.}

\medskip

\begin{proof}
We use the same algorithm as in \Cref{thm:diameterListsimple} but speed up the derandomization step (cf. \Cref{lem:derand}) by fixing $\Omega(\log n)$ bits of the random seed in $O(1)$ rounds. That is, we fix one bit of each node's candidate color in a constant number of rounds. Assume the bits $\{1,\dots,\ell-1\}$ are already fixed and we want to fix the $\ell$-th bit. Instead of first using Linial's algorithm to compute an $O(\Delta^2)$ coloring we simply use the unique Ids of the nodes as an input coloring which yields a seed length of $O(\log n)$. To fix the $\ell$-th bit of all prefixes, we split the seed into $\bigO{1}$ segments of length $\lambda\leq\log n$ and choose a \textit{partial seed} $S\in\{0,1\}^{\lambda}$ for each segment. For the first segment, a leader chooses a subset $V'\subseteq V$ of size $2^{\lambda}$ and a bijection $\mathcal{R}:V'\to\{0,1\}^{\lambda}$ and sends $\mathcal{R}(v)$ to node $v$. We say that $v$ is responsible for $\mathcal{R}(v)$. The nodes exchange the tuples $(v,\mathcal{R}(v))$ such that each node knows which node is responsible for which partial seed. Next, each $u\in V$ learns $k_0(v)$ (the number of colors in $L_{\ell-1}(v)$ starting with $0$) and $k_1(v)$ from each of its neighbors $v$. With this information, $u$ can compute $\ev[R_1\circ\ldots\circ R_{\lambda}=\mathcal{R}(v)]{\phi_{\ell}(u)}$ for any $v\in V'$ and sends this value to $v$. Now each node $u\in V'$ can compute $\sum_{v\in V}\ev[R_1\circ\ldots\circ R_{\lambda}=\mathcal{R}(u)]{\phi_{\ell}(v)}$ and sends it to the leader. The leader chooses the partial seed minimizing this sum and broadcasts it to every node. This way we proceed with all other segments and thus the derandomization of the whole seed is done in $\bigO{1}$ rounds. It follows that coloring a constant fraction of the nodes (cf. proof of \Cref{thm:colorhalf}) can be done in $\bigO{\log C}$ rounds.

After a constant number of iterations of this procedure, i.e., after $\bigO{\log C}$ rounds, at least half of the nodes are colored. This implies that we can speed up the following iterations by fixing \emph{two} bits of each node's candidate color in $O(1)$ rounds. More generally, if the number of uncolored nodes is at most $\frac{n}{2^i}$, we can fix $i$ bits of each node's candidate color in $O(1)$ rounds and thus color a constant fraction of the nodes in time $O\left(\frac{\log C}{i}\right)$. This follows from the fact that if at most $\frac{n}{2^i}$ nodes are left, we have $\Delta\leq\frac{n}{2^i}$ ($n$ is the number of nodes in the original communication network and $\Delta$ the maximum degree in the subgraph induced by the uncolored nodes) and by using Lenzen's routing algorithm, each node can send $n/\Delta\geq2^i$ values to each of its neighbors in $O(1)$ rounds. E.g., for $i=2$, each node $u$ sends $k_{00}(u)$ (the number of colors in $L_{\ell-1}(u)$ starting with $00$), $k_{01}(u)$, $k_{10}(u)$ and $k_{11}(u)$ to each of its neighbors. By adapting the derandomization process in \Cref{lem:derand} in a straightforward fashion, we can fix the next two bits of each node's candidate color.

As a result, it takes \[O\left(\sum_{i=1}^{\log\Delta}\frac{\log C}{i}\right)=O\left(\log\log\Delta\log C\right)\] rounds until the number of uncolored nodes is reduced to $n/\Delta$. Then we can use Lenzen's routing algorithm \cite{Lenzen13} to send in $O(1)$ rounds the subgraph of uncolored nodes to a leader which locally solves the problem.
\end{proof}
Next, we turn to coloring in the MPC model and we first formally explain how a $(degree+1)$-list-coloring instance is given in the MPC model. 
 
\noindent\textbf{\boldmath$(degree+1)$-list-coloring instance in MPC.} Given a graph with $n$ nodes and $m$ edges, we assume that each machine has a local memory of $S=\Theta\left(n^{\alpha}\right)$ for some constant $\alpha>0$ and we have $\Theta\left(\frac{m+n}{S}\right)$ machines.\footnote{The constant $\alpha$ cannot be chosen by the algorithm designer but is determined by the system. However, the algorithm designer can choose the constant in the $\Theta(\cdot)$ for $S$ and for the number of machines.} Each edge $\{u,v\}$ is stored as $\{\text{Id}_u,\text{Id}_v\}$ on some machine. For the $(degree+1)$-list-coloring problem, for each $u$ and each color $c$ from $u$'s list there is a \emph{list entry} $(\text{Id}_u,c)$ stored on some machine. Both, edges and list entries, can be distributed adversarially on the machines.

In the following theorems, we further assume that we have directed edges $(u,v)$ and $(v,u)$ for each edge $\{u,v\}\in E$ stored on the machines and the edges and list entries are sorted in lexicographic order (cf. \Cref{def:sorting}). Further we assume that the machine storing $(u,v)$ knows on which machine $(v,u)$ is stored. This can be achieved in $O(1)$ in the following way: If machine $i$ stores an edge $\{u,v\}$, it replaces the edge by the triples $(u,v,i)$ and $(v,u,i)$. We sort these triples as well as the list entries $(u,c)$ in lexicographic order in $O(1)$ rounds using the algorithm from \Cref{lem:MPCbasics}. Assume that after the sorting, $(u,v,i)$ is stored on machine $j$ and $(v,u,i)$ is stored on machine $k$. Via communication over machine $i$, machine $j$ learns that $(v,u,i)$ is stored on machine $k$ and vice versa.
In \Cref{sec:MPCbasics} we formally show that seemingly trivial procedures such as deleting colors chosen by a neighbor from a list that is stored in a distributed manner can be performed in $O(1)$ rounds in the MPC model. We also provide communication primitives such as aggregation trees for the nodes of the graph and we freely use these in our proofs without explicitly mentioning them. For the details we refer to \Cref{sec:MPCbasics}.
\begin{observation}
\label{obs:MPCDeltaToList}
We can reduce the $(\Delta+1)$-coloring problem where initially no color lists are given to the $(degree+1)$-list-coloring problem in the following way: Let $v_1,\dots,v_{\deg(u)}$ be $u$'s neighbors sorted by increasing Id. Each machine storing an edge $(u,v)$ learns $v$'s position within the list of $u$'s neighbors, i.e., the $i\in\{1,\dots,\deg(u)\}$ for which $v=v_i$, in $O(1)$ rounds (cf. \Cref{cor:rank}) and writes the list entry $(u,i)$ to its memory (w.l.o.g. we may assume that so far, each machine only used half of its memory). The machine storing $(u,v_{\deg(u)})$ produces both list entries $(u,\deg(u))$ and $(u,\deg(u)+1)$. This way we produced a color list $L(u)\subseteq[\Delta+1]$ of size $\deg(u)+1$ for each node $u$. Thus all our $(degree+1)$-list-coloring results also apply to the standard $(\Delta+1)$-coloring problem. 
\end{observation}
\medskip

The algorithm in \Cref{thm:colorhalf} computes an MIS on constant degree subgraphs when every node has chosen a candidate color. To keep our MPC coloring algorithm self-contained and independent from the implementation of an MIS algorithm we explain how to avoid the computation of an MIS. 

\noindent\textbf{How to Avoid MIS.} We adapt the algorithm from \Cref{thm:colorhalf} such that it is not necessary to compute an MIS at the end. To this end, we produce the coins in \Cref{lem:derand} with higher accuracy. Concretely, we choose $\varepsilon=\frac{1}{10\Delta(\Delta+1)\lceil\log C\rceil}$, i.e., we add a $1/(\Delta+1)$ factor in the accuracy. It follows that after fixing one color-bit of all nodes, the sum of all potentials increased by at most $10\varepsilon\Delta n=\frac{n}{(\Delta+1)\lceil\log C\rceil}$ (cf. \Cref{lem:expectslack} and \Cref{lem:derand}---increasing the accuracy this way increases the runtime in \Cref{lem:derand} only by a constant factor). As initially we have
\begin{equation}\label{eq:potbound_MPC}
\sum_{v\in V}\Phi_0(v)=\sum_{v\in V}\frac{\deg(v)}{|L(v)|}<\sum_{v\in V}\frac{|L(v)|-1}{|L(v)|}\stackrel{(*)}{\leq} n-\frac{n}{\Delta+1},
\end{equation}
it follows that after fixing all $\ell=\lceil\log C\rceil$ bits we have
\begin{align}\label{eq:potincrease}
\sum_{v\in V}\Phi_{\ell}(v)\leq\sum_{v\in V}\Phi_0(v)+\frac{n}{\Delta+1}<n~.
\end{align}
For the last inequality in \ref{eq:potbound_MPC} (marked by $(*)$) we used $|L(v)|\leq\Delta+1$. To ensure this property for each node $v$ at each stage of the algorithm, we must delete colors from $v$'s list as soon as $v$'s degree decreased by more than $v$'s list size. This deletion can be done in $O(1)$ rounds in the MPC model with similar techniques as in \Cref{obs:MPCDeltaToList}.

We obtain that at least half of the nodes have $\Phi_{\ell}<2$, i.e., $\Phi_{\ell}\leq1$ which means that at least half of the nodes have at most one neighbor that chose the same candidate color. Computing an MIS on the graph induced by these nodes can be done in $1$ round by letting the vertex with the larger Id join the MIS; then all MIS nodes keep their color permanently.

\medskip

We first show how to solve $(\mathit{degree}+1)$-list-coloring with linear memory and afterwards move to the more involved sublinear memory regime. 

\medskip

\noindent\textbf{\Cref{cor:MPClin}}
\textit{There is a deterministic MPC algorithm that solves the $(\mathit{degree}+1)$-list-coloring problem in $O(\log\Delta\log C)$ rounds with linear memory.}

\begin{proof}
We proceed similarly as in the CONGESTED CLIQUE model, i.e., we fix one bit of each node's candidate color in a constant number of rounds by fixing $\Omega(\log n)$ bits of the random seed in $O(1)$ rounds in the derandomization. Assume the bits $\{1,\dots,\ell-1\}$ are already fixed and we want to fix the $\ell$-th bit. We use the Ids of the nodes as an input coloring which yields a seed length of $O(\log n)$. As all edges $(u,v)$ for $v\in\nh(u)$ and list entries $(u,c)$ fit on one machine, we can achieve that they are stored on the same machine which we name $\mathcal{M}_u$ (note that $\mathcal{M}_u=\mathcal{M}_v$ for $u\neq v$ is possible).

For the derandomization, we split the seed into $\bigO{1}$ segments of length $\lambda\leq\frac{\log n}{2}$. If $\mathcal{M}_u$ stores the edge $(u,v)$, it sends the values $k_1(u)$ and $|L_{\ell-1}(u)|$ to $\mathcal{M}_v$. We obtain that each machine $\mathcal{M}_u$ knows the values $k_1(v)$ and $|L_{\ell-1}(v)|$ for all $v\in\nh(u)$ and hence can compute a vector $\mathcal{V}_u$ of length $2^{\lambda}\leq\sqrt{n}$ with entries $\ev[R_1\circ\ldots\circ R_{\lambda}=\mathcal{R}]{\Phi_{\ell}(u)}$ for any partial seed $\mathcal{R}\in\{0,1\}^{\lambda}$ (we assume each machine knows how to compute a node's coin from a seed, cf. \Cref{eq:coins}). For derandomizing the first segment of the seed, we need to compute $\sum_{v}\mathcal{V}_v$ (where we sum the vectors componentwise) and then choose the partial seed whose entry is the smallest. For this purpose we build an aggregation tree with degree $\sqrt{n}$ and depth $O(1)$ (cf. \Cref{lem:MPCbasics}). Each machine sums up the vectors it has and sends it to its parent. As a machine has at most $\sqrt{n}$ children each sending a vector of length at most $\sqrt{n}$, a machine receives at most $n$ values. Finally, the root machine learns $\sum\mathcal{V}_v$, chooses the partial seed with minimum entry and broadcasts it to each machine. Afterwards we proceed with the next seed segment.

After fixing all $\log C$ bits, each node has chosen a candidate color and a constant fraction of the nodes can keep their color permanently. If node $u$ has permanently chosen a color, $\mathcal{M}_u$ informs $\mathcal{M}_v$ about it (for each $v\in\nh(u)$) which updates $v$'s color list.

After $O(\log\Delta)$ iterations of this procedure, i.e., after $O(\log\Delta\log C)$ rounds, there are at most $n/\Delta^2$ uncolored nodes left. It follows that the graph induced by the uncolored nodes has at most $n/\Delta$ edges, which means that it fits on a single machine together with all remaining color lists (with length $O(\Delta)$ each). Thus all machines can send their edges and color lists to a leader machine which locally list-colors the remaining graph.
\end{proof}

\noindent\textbf{\Cref{cor:MPCsub}}
\textit{There is a deterministic MPC algorithm that solves the $(\mathit{degree}+1)$-list-coloring problem in $O(\log\Delta\log C+\log n)$ rounds with sublinear memory.}

\begin{proof}
We proceed similarly as in the linear memory regime, i.e., we fix one bit of each node's candidate color in a constant number of rounds. We take the Ids of the nodes as input coloring and obtain a seed length of $O(\log n)$. With sublinear memory it is not possible to store all edges incident of a node $u$ together with all colors in $u$'s list on a single machine, but we store it on a set of machines that are connected via an aggregation tree. We call a machine storing an edge $(u,v)$ for a $v\in\nh(u)$ or a list entry $(u,c)$ for a $c\in L(u)$ a \textit{$u$-machine}. For each node $u$ we build an aggregation tree with degree $\sqrt{n}$ and depth $O(\alpha^{-1})$ with the $u$-machines as leafs (cf. \Cref{lem:MPCbasics}).

Next we describe the derandomization. Assume the bits $\{1,\dots,\ell-1\}$ are already fixed and we want to fix the $\ell$-th bit. We split the seed into $O(1/c)$ parts of length $\lambda\leq c\log n$ for a sufficiently small constant $c<\alpha$. We do an edge-based computation of the sum of all potentials. As in \Cref{lem:expectation} we have \[\ev{\sum_{v\in V}\Phi_{\ell}(v)}=\ev{\sum_{e\in E_{\ell-1}}X_e}\text{\quad with \quad}X_e=\ind_{e\in E_{\ell}}\left(\frac{1}{|L_{\ell}(u)|}+\frac{1}{|L_{\ell}(v)|}\right)~.\]
Each $u$-machine learns the values $k_1(u)$ and $|L_{\ell-1}(u)|$ in $O(1)$ rounds via communication over the aggregation tree. If a machine stores the edge $(u,v)$ it sends $k_1(u)$ and $|L_{\ell-1}(u)|$ to the machine storing $(v,u)$. We obtain that the machine storing the edge $e=(u,v)$ has the values $k_1(u)$, $|L_{\ell-1}(u)|$, $k_1(v)$ and $|L_{\ell-1}(v)|$ and hence can compute a vector $\mathcal{V}_e$ of length $2^{\lambda}\leq n^c$ with entries $\ev[R_1\circ\ldots\circ R_{\lambda}=\mathcal{R}]{X_e}$ for any partial seed $\mathcal{R}\in\{0,1\}^{\lambda}$. We connect the aggregation trees of the nodes to one tree with degree $\sqrt{S}$ and depth $O(\alpha^{-1})$. Via aggregation over this tree the root machine learns $\sum_{e\in E_{\ell-1}}\mathcal{V}_e$, chooses the partial seed with minimum entry and broadcasts it to each machine. Afterwards we proceed with the next seed segment. After fixing all $\log C$ bits, each node has chosen a candidate color and a constant fraction of the nodes can keep their color permanently. Afterwards, the machines update their color lists. That is, using the set difference algorithm from \Cref{lem:MPCbasics}, the machine storing $(u,c)$ can learn whether $c$ has been taken by one of $u$'s neighbors and deletes it if this is the case.

If $\Delta>n^{\alpha/2}$, we apply this $O(\log C$) rounds procedure $O(\log n)$ times, each time coloring a constant fraction of the nodes, yielding an $O(\log n\log C)=O(\log \Delta \log C)$ algorithm.  If $\Delta<n^{\alpha/2}$, we apply this procedure $O(\log\Delta)$ times, taking $O(\log\Delta\log C)$ rounds. Afterwards, the number of uncolored nodes is reduced to $n/\Delta^2$ and we can list-color the remaining nodes in $O(\log n)$ rounds using \Cref{lem:mpcfast}.
\end{proof}

\begin{lemma}\label{lem:mpcfast}
Consider the sublinear memory regime of the MPC model where each machine has memory $S=\Theta(n^{\alpha})$ for some $\alpha>0$. Let $G$ be a graph with $\Delta<n^{\alpha/2}$ and assume we have a total memory of $\Omega(n\Delta^2)$. Then we can solve list-coloring on $G$ in $O(\log n)$ rounds.
\end{lemma}

\begin{proof}
We take again the node Ids as initial coloring and obtain a seed length of $O(\log n)$. However, in the given setting, it is not necessary to fix each nodes candidate color bitwise. Instead each node directly chooses a color from its list uniformly at random if the seed is chosen uniformly at random. Expressed in the framework of \Cref{lem:derand} this means that given a random seed, each node chooses a color from its list by choosing one bit after another, producing the coins for each bit from the seed as before. The derandomization of this process goes along similar lines as the proof of \Cref{cor:MPCsub}. That is, we have \[\ev{\sum_{v\in V}\Phi_{\lceil\log C\rceil}(v)}=\ev{\sum_{e\in E}X_e} \text{ with } X_e=\ind_{e\in E_{\lceil\log C\rceil}}~.\]
To compute $\ev{X_e}$ for $e=(u,v)$ conditioned on some partly chosen seed, the machine storing $e$ needs the color lists of $u$ and $v$. We have at most $n\Delta$ edges and $\Delta<n^{\alpha/2}$. As we have a global memory of $\Omega(n\Delta^2)$ and each machine has memory $\Theta(n^{\alpha})$, we can achieve that there is only one $u$-machine for each node $u$ (that is, all edges outgoing from $u$ and the colors from $L(u)$ are on one machine) and each machine has $\Omega(\Delta)$ extra space for each edge it stores. The machine storing edge $(u,v)$ can therefore send $L(u)$ (which has size $O(\Delta)$) to the machine storing $(v,u)$. Now each machine knows the color lists $L(u)$ and $L(v)$ for each edge $e=(u,v)$ it stores and can compute $\ev{X_e}$. The procedure of finding a good seed goes along similar lines as described in the proof of \Cref{cor:MPCsub}.

The increase of the potential can be upper bounded as given in \Cref{eq:potincrease} and hence a constant fraction of the nodes can keep their candidate color. The procedure described so far runs in $O(1)$ rounds and hence the list-coloring takes $O(\log n)$ rounds.
\end{proof}

%% file: sec5-MPCbasics.tex
\section{Basic MPC Tools}
\label{sec:MPCbasics}

In this section, we provide a set of basic algorithms and constructions that can be carried out deterministically in constant time in the sublinear MPC model. The algorithms are used as subroutines in our sublinear memory MPC algorithms. In the following, when we say that an MPC algorithm is given a set or a multiset of $N$ values, we assume that each of the values is initially given to an arbitrary machine. Throughout this section, we further assume that there is a parameter $S=N^{\alpha}$ for some constant $\alpha>0$ such that every machine has space for $c\cdot S$ values for a sufficiently large constant $c>0$. We will assume that each machine stores at most $S$ values such that there is enough space on the machine to store a constant amount of additional data per stored value during a computation. We note that in the MPC model, it is standard to assume that the global memory is by a sufficiently large constant factor $c>0$ larger than the total size of the input. We can then always store the input in such a way that only a (sufficiently small) constant fraction of each machine is occupied. We will therefore also generally assume that the number of available machines is $c'\cdot N/S$ for a sufficiently large constant $c'$. We consider the following three basic problems. 

\begin{definition}[Sorting]\label{def:sorting}
  Assume that an MPC algorithm receives a multiset of $N$ values from a totally ordered set as input. The MPC algorithm is said to sort the input values if at the end machine $i\in\set{1,\dots,\lceil N/S\rceil}$ stores the input values with ranks $(i-1)S+1,\dots,iS$ in the sorted order of all the $N$ input values.
\end{definition}

\begin{definition}[Prefix Sums]\label{def:prefixsum}
  Assume that an MPC algorithm receives a set $X$ of $N$ values from a totally ordered set $\calX$ as input. Assume further that there is an associative binary operation $\oplus:\calX\times\calX \to \calX$. Assume that $x_i$ is the value of the input value with rank $i$ in the sorted order of $X$. The MPC algorithm is said to solve the prefix sums problem w.r.t.\ $\oplus$ if at the end for every $i\in \set{1,\dots, N}$, the machine holding value $x_i$ also holds the value $s_i := \bigoplus_{j=1}^i x_i$.
\end{definition}

\begin{definition}[Set Difference]\label{def:setdiff}
  Assume that an MPC algorithm receives a collection of sets $A_1,\dots,A_k$ and a collection of multisets $B_1,\dots,B_k$ as inputs. The MPC algorithm is said to solve the set difference problem if at the end for every $i\in \set{1,\dots,k}$ and for every $a\in A_i$, the machine holding $a$ knows if $a$ is contained in $B_i$.
\end{definition}

The last problem defines a tree structure that is useful to carry out computations on a set or on a collection of sets.

\begin{definition}[Aggregation Tree Structure]\label{def:aggregationtrees}
  Assume that an MPC algorithm receives a collection of sets $A_1,\dots,A_k$ with elements from a totally ordered domain as input. In an \emph{aggregation tree structure} for $A_1,\dots,A_k$, the elements of $A_1,\dots,A_k$ are stored in lexicographically sorted order (they are primarily sorted by the number $i\in\set{1,\dots, k}$ and within each set $A_i$ they are sorted increasingly). For each $i\in \set{1,\dots,k}$ such that the elements of $A_i$ appear on at least $2$ different machines, there is a tree of constant depth containing the machines that store elements of $A_i$ as leafs and where each inner node of the tree has at most $\sqrt{S}$ children. The tree is structured such that it can be used as a search tree for the elements in $A_i$ (i.e., such that an in-order traversal of the tree visits the leaves in sorted order). Each inner node of these trees is handled by a seperate additional machine. In addition, there is a constant-depth aggregation tree of degree at most $\sqrt{S}$ connecting all the machines that store elements of $A_1,\dots,A_k$.
\end{definition}

\begin{lemma}\label{lem:MPCbasics}
  There are constant-time MPC algorithms to solve the \emph{sorting}, \emph{prefix-sums}, and \emph{set difference} problems, as well as to compute an \emph{aggregation tree structure} for a given collection of input sets.
\end{lemma}
\begin{proof}
  The constant-time sorting algorithm follows directly from the constant-time sorting algorithm in the MapReduce model that is described by Goodrich et al.~\cite{goodrich11}. In \cite{goodrich11}, the authors consider an I/O-bound MapReduce model, which can be simulated efficiently by the MPC model (the two models are essentially equivalent). Theorem 3.1 in \cite{goodrich11} gives an efficient simulation of BSP algorithms in the I/O-bound MapReduce model. In combination with the BSP sorting algorithm in \cite{goodrich99}, this yields the desired sorting algorithm in the MapReduce and thus also in the MPC model.

  The constant-time prefix-sums algorithm also follows form \cite{goodrich11}. In \cite{goodrich11}, it is described how to compute the prefix sums in constant time if each machine only has one value. Note however that if we first locally compute the sum of all values of each machine and afterwards compute the prefix sums of these values, the prefix sums for all the elements can easily be derived in a constant number of additional rounds.

  As the third part of the proof, we show how to construct an aggregation tree structure in constant time. We first sort the elements of the sets $A_1,\dots,A_k$ lexicographically (first by the set index $i\in\set{1,\dots,k}$ and then within each set $A_i$). As described above, this can be done in $O(1)$ time in the MPC model. As a next step, for every $i\in\set{1,\dots,k}$ and each element $a\in A_i$, we make sure that the machine storing $a$ learns the cardinality $|A_i|$ of the set $A_i$ and that the machine learns the position of $a$ within the sorted order of the elements of $A_i$. Note that as soon as this is done, every machine that stores elements of a set $A_i$ knows exactly where the other elements of $A_i$ are stored and therefore building the aggregation trees for each $A_i$ is straightforward.

  To compute the positions of the elements within the sets $A_i$, we use the prefix sums algorithm discussed above. Assume that all the elements of $A_i$ are from a globally ordered domain $\calX$. We first define an associative binary relation $\oplus: (\mathbb{N}\times \calX)\times (\mathbb{N}\times \calX) \to (\mathbb{N}\times \calX)$ as follows. For $a_1,a_2\in \mathbb{N}$ and $x_1,x_2\in \calX$, we define
  \[
  (a_1,x_1)\oplus(a_2,x_2) :=
  \begin{cases}
    (a_1, x_1 + x_2) & \text{if $a_1=a_2$}\\
    (a_1, x_1) & \text{if $a_1 > a_2$}\\
    (a_2, x_2) & \text{if $a_2 > a_1$}
  \end{cases}\quad.
  \]
  It is easy to check that the binary operation $\oplus$ is associative. For each element $a\in A_i$, we now build the tuple $(i, a)$ (on the machine where $a$ is stored) and we apply the prefix sums discussed above to these tuples w.r.t.\ the binary operation $\oplus$. Because the elements of all the sets $A_1\cup\dots\cup A_k$ are sorted by increasing set index $i\in \set{1,\dots,k}$, this prefix sums operation gives exactly what we need. For each set $A_i$, it numbers the elements in sorted order from $1$ to $|A_i|$. To determine the size of each $A_i$, we can use an analogous algorithm to number the elements of each set $A_i$ in the reverse direction. If an element $a\in A_i$ knows its position in $A_i$ from both sides, it also directly knows the size of $A_i$. This therefore shows that the aggregation trees for each set can be computed in constant time in the MPC model. The aggregation tree connecting all the machines that store elements of $A_1,\dots,A_k$ can be done in the same way by treating the collection of the elements of all the sets as one large set.

  It remains to show that also the set difference operation can be carried out in $O(1)$ MPC rounds. First note that we can w.l.o.g.\ assume that $k=1$. That is, we are only given a single set $A$ and a single multiset $B$ for each $a\in A$, we need to determine whether $a\in B$. To reduce the general problem to the problem with a single set $A$ and a single multiset $B$, we define $A:=\bigcup_{i=1}^k\bigcup_{a\in A_i} (i, a)$ and  $B:=\bigcup_{i=1}^k\bigcup_{b\in B_i} (i, b)$. As a next step, we compute an aggregation tree structure for $A$ and $B$. Note that while $B$ is a multiset, formally the aggregation tree structure is only defined for sets. For the construction, we can easily turn $B$ into a set by arbitrarily (but uniquely) labeling the elements of $B$ (e.g., by only using a single copy of each element of $B$ per machine and adding the machine number as a label to the elements). Further, after bulding the aggregation tree for $B$, we can efficiently (i.e., in constant time) remove all except one copy for each element in $B$ so that in the following, we can w.l.o.g.\ assume that $B$ is a set. In the following, we will refer to the aggregation tree for set $A$ as the $A$-tree and we similarly refer to the aggregation tree for set $B$ as the $B$-tree.

  Note that because of the search tree property of the aggregation trees of $A$ and $B$, for a given element of $a\in A$ it is straightforward to determine in constant time if $a\in B$ and similarly for an element $b\in B$ one can easily check in constant time if $b\in A$. The challenge is to efficiently do these searches for all elements in parallel. To describe this parallel search, we first define some notation. Let $\calM_A$ be the set of machines that participate in the $A$-tree (i.e., the machines storing the elements of $A$ and the machines handling the inner nodes of the aggregation tree). Similarly, let $\calM_B$ be the set of machines that participate in the $B$-tree. For simplicity, assume that both aggregation trees have height exactly $h$ for some $h=O(1)$ and that all the leaf nodes (i.e., the machines storing the elements of $A$ and $B$) are on level $h$ of the respective tree. This is easy to guarantee by adding some dummy nodes where necessary. For every $\ell\in \set{0,\dots,h}$, let $\calM_A^{(\ell)}\subseteq \calM_A$ and $\calM_B^{(\ell)}$ be the subsets of the machines in $\calM_A$ and $\calM_B$ that are on level $\ell$ of the respective tree (where level $0$ consists of the root nodes of the two trees). For each machine $M$ in $\calM_A$ or $\calM_B$, let $x_{\min}(M)$ and $x_{\max}(M)$ be the smallest and largest elements stored in the subtree of $M$. 

  As a first step, we compute $M_B[M_A]\in \calM_B\cup \set{\bot}$ for each machine $M_A\in \calM_A$ in such a way that if $M_B[M_A]=\bot$, there are no elements of $A\cap B$ stored in the subtree of $M_A$ of the $A$-tree and otherwise, (among other things) the subtree of $M_B[M_A]$ of the $B$-tree contains all the values of $A\cap B$ that are stored in the subtree of $M_A$ of the $A$-tree. For a fixed $M_A\in \calM_A$, the machine $M_B[M_A]$ is determined by using the following iterative process. We initially set $M_B[M_A]\in\calM_B^{(0)}$ to be the root machine of the $B$-tree. Assume that we have currently set $M_B[M_A]$ such that $M_B[M_A]\in \calM^{(\ell)}$ for some $\ell\in\set{0,\dots,h}$. If $\ell=h$ (i.e., $M_B[M_A]$ is a leaf machine), we assign $M_B[M_A]$ to $M_A$. Otherwise, let $M_{B,1},\dots,M_{B,k}\in\calM_B^{(\ell+1)}$ be the direct children of $M_B[M_A]$ in the $B$-tree.  If there is exactly one such child machine $M_{B,i}$ such that the value ranges of $M_{B,i}$ and $M_A$ intersect (i.e., such that $x_{\max}(M_A)\geq x_{\min}(M_B[M_A])$ and $x_{\min}(M_A)\leq x_{\max}(M_B[M_A])$), we set $M_B[M_A]:=M_{B,i}$ and continue. Otherwise, if there is no machine $M_{B,i}$ for which the ranges of $M_A$ and $M_{B,i}$ intersect, we define $M_B[M_A]:=\bot$ and if there are at least $2$ child machines $M_{B,i}$ for which the value range intersects with the value range of $M_A$, we assign $M_B[M_A]$ to $M_A$. Note that the construction of $M_B[M_A]$ immediately implies that all the values of $B\cap A$ that are stored in the subtree of $M_A$ in the $A$-tree are stored in the subtree of $M_B[M_A]$ in the $B$-tree (and if $M_B[M_A]=\bot$, it implies that the intersection of the values stored in the subtree of $M_A$ and $B$ is empty).

  We next show that the machines $M_B[M_A]\calM_B$ for all machines $M_A\in\calM_A$ can be done in constant time in the MPC model. Clearly, if $M_A'$ is a machine in the subtree of $M_A$ (in the $A$-tree) then also $M_B[M_A]$ is a machine in the subtree of $M_B[M_A]$ (in the $B$-tree). The assignment can therefore be computed in a top-down fashion. It remains to show that the computation does not exceed the communication budget of each MPC machine. In addition to computing the assignment, we make sure that for each machine $M_A\in\calM_A$, if $M_B[M_A]$ is not a leaf machine of the $B$-tree, $M_A$ also learns the complete state of $M_B[M_A]$. Note that the since the degree of both trees is at most $\sqrt{S}$, the state of $M_B[M_A]$ consists of $\sqrt{S}$ values (the information about the $\sqrt{S}$ children and their value ranges). Therefore, when moving down the $A$-tree from a node $M$ to its at most $\sqrt{S}$ children, the information about $M_B[M]$ can be copied to the children of $M$ in a single MPC round. Hence, when computing the machine $M_B[M_A]$ assigned to $M_A$, we can assume that $M_A$ already knows the state of the machine $M_B[M]$ for the parent machine $M$ of $M_A$ (if $M_A$ is not the root node of the $A$-tree). If $M_A$ is the root machine of the $A$-tree, it is clear that it can compute the machine $M_B[M_A]$ and learn its state in constant time. Otherwise, since $M_A$ already knows the state of $M_B[M]$, it can then locally check if there is exactly one subtree of $M_B[M]$ that intersects with the value range of $M_A$. If this is not the case, we have $M_B[M_A]=\bot$ or $M_B[M_A]=M_B[M]$ and we are done. Otherwise, let $M_{B,i}$ be this unique subtree of $M_B[M]$ that intersects with the range of $M_A$. Machine $M_A$ then asks machine $M_{B,i}$ for its state. In order to show that we do not exceed the bandwidth requirement, we have to show that the state of each machine $M_{B,i}$ is required by at most $O(\sqrt{S})$ different $A$-tree machines. In the following, assume that $M_B$ is the parent machine of $M_{B,i}$. Assume that $M_A\in\calM_A^{(\ell)}$ for some level $\ell\geq 1$. Machine $M_A$ only asks for the state of $M_{B,i}$ if the range of the parent machine $M\in\calM_A^{(\ell-1)}$ of $M_A$ intersects with $M_{B,i}$ and with at least one other subtree $M_{B,j}$ of $M_B$ (i.e., it either also intersects fith $M_{B,i-1}$ or with $M_{B,i+1}$). Note however that there can be at most one machine in $\calM_A^{(\ell-1)}$ for which the value range intersects with $M_{B,i-1}$ and with $M_{B,i+1}$ and there can be at most one such machine for which the value range intersects with $M_{B,i}$ and with $M_{B,i}$. Only the $O(\sqrt{S})$ children machines of these two machines can ask for the state of $M_{B,i}$. The assignment $M_B[M_A]$ for all $M_A\in\calM_A$ can therefore be computed in constant time in the MPC model.

  We next show how we can use the assignment $M_B[M_A]$ to the machines $M_A\in\calM_A$ to determine the intersection of $A$ and $B$. Recall that each leaf machine $M_A\in\calM_A^{(h)}$ needs to determine which of its at most $S$ stored values of $A$ are contained in $B$. Consider some machine $M_A\in\calM_A^{(h)}$. If $M_B[M_A]=\bot$, we $M_A$ knows that there is no intersection between the values of $A$ it stores and $B$. Let us therefore assume that $M_B[M_A]\neq\bot$. We first consider the case where $M_B[M_A]$ is an inner node of the $B$-tree (i.e., $M_B[M_A]\in\calM_B^{(\ell)}$ for some $\ell<h$). Let $M_{B,1},\dots,M_{B,k}$ be the direct children of $M_B[M_A]$ in the $B$-tree. Recall that $M_A$ knows the state of $M_B[M_A]$ and it thus knows $M_{B,1},\dots,M_{B,k}$ and $x_{\min}(M_{B,i})$ and $x_{\max}(M_{B,i})$ for each $i\in\set{1,\dots,k}$. We also know that the value range of $M_A$ intersects with the ranges of at least two of the machines $M_{B,1},\dots,M_{B,k}$ as otherwise, we would have set $M_B[M_A]$ differently. Therefore, for each $M_{B,i}$, there are at most $2$ leaf machines $M_A\in\calM_A^{(h)}$ of the $A$-tree for which the value range intersects with the value range of $M_{B,i}$. The machine $M_A$ can therefore send all its stored elements of $A$ that fall within the range of machine $M_{B,i}$ for each $i\in \set{1,\dots,k}$. Note that while $M_A$ might have to communicate with up to $\sqrt{S}$ different machines, the total number of values it has to send is at most $S$. Since each such machine $M_{B,i}$ is contacted by at most $2$ machines from the $A$-tree, it can in constant time determine which of the elements it receives are stored in its subtree and report this information back to the respective $A$-tree machines. Hence, if $M_B[M_A]$ is an inner node of the $B$-tree, $M_A$ can efficiently learn its part of $A\cap B$. It therefore remains to consider the case where $M_B[M_A]$ is a leaf machine of the $B$-tree. Now, all the elements in the intersection between the elements stored at $M_A$ and the elements in $B$ are stored on machine $M_B[M_A]$ in the $B$-tree. However, we cannot bound the number of leaf machines $M_A$ that are assigned the same $B$ machine $M_B[M_A]$. The intersection between the elements stored at $M_A$ and the elements stored at $M_B[M_A]$ can therefore not be computed by direct communication betweek the two machines. Instead, we delegate this task up the $A$-tree. Let $M$ be the highest (i.e., closest to the root) ancestor of $M_A$ in the $A$-tree for which $M_B[M]=M_B[M_A]$. Instead of $M_A$, machine $M$ asks machine $M_B$ for its elements in the range between $x_{\min}(M)$ and $x_{\max}(M)$. After receiving these values, $M$ can efficiently propagate them down the $A$-tree by always only forwarding the values that fall within the range of each particular child. Thus, after $M$ knows the values of $M_B$ in the range between $x_{\min}(M)$ and $x_{\max}(M)$, $M_A$ can efficiently learn its intersection with $B$. In order to show that the algorithm works, it only remains to show that the leaf machine $M_B$ of the $B$-tree does not get too many requests from machines $M'\in\calM_A$ for which $M_B[M']=M_B$. Note however that on each level $\ell$ (i.e., in each set $\calM_A^{(\ell)}$, there can be at most $O(\sqrt{S})$ machines $M'$ for which $M_B[M']=M_B$ and such that for the parent $M''$ of $M'$, we have $M_B[M'']\neq M_B$. This concludes the proof.
\end{proof}

In the construction of the aggregation tree structure, we proved the following as an intermediate result:

\begin{corollary}\label{cor:rank}
There is a constant-time MPC algorithm that given a collection of sets $A_1,\dots,A_k$ with elements from a totally ordered domain, for all $i\in\{1,\dots,k\}$ and $a\in A_i$ the machine storing $a$ learns the rank of $a$ within $A_i$.
\end{corollary}
